\documentclass[a4paper,12pt]{article}

\usepackage{mathbbol}
\usepackage{mathrsfs}
\usepackage{graphicx}
\usepackage{geometry}
\usepackage{amsmath}
\usepackage{amsfonts}
\usepackage{amssymb}
\usepackage{amsthm}
\usepackage{colortbl}
\usepackage{epstopdf}
\usepackage{xcolor}
\usepackage{bm}
\usepackage{float}
\usepackage{enumerate}
\usepackage[normalem]{ulem}
\usepackage{placeins}
\usepackage{braket}
\usepackage{comment}
\usepackage{cancel}
\usepackage{subfig}
\usepackage{empheq}

\newcommand{\bvec}[1]{\mathbf{#1}}

\newcommand{\vr}{\bvec{r}}

\newcommand\Z{\mathbb{Z}}
\newcommand\R{\mathbb{R}}
\newcommand\N{\mathbb{N}}
\newcommand\C{\mathbb{C}}
\newcommand\Sc{\mathscr{S}}

\newcommand{\ldos}{\mathscr{D}_n}
\newcommand{\dd}{~{\rm d}}
\newcommand{\dist}{R}

\newcommand{\Ntr}{L}
\newcommand{\Nit}{N_{\rm i}}
\newcommand{\ham}{\mathcal{H}}
\newcommand{\ylm}{Y_{\ell m}(\hat{\bm{r}})}

\newtheorem{theorem}{Theorem}[section]
\newtheorem{lemma}{Lemma}[section]
\newtheorem{remark}{Remark}[section]

\renewcommand{\theequation}{\arabic{section}.\arabic{equation}}

\renewcommand{\thefigure}{\arabic{section}.\arabic{figure}}

\title{Approximations of the Green's Function in
\\[1ex]
Multiple Scattering Theory for Crystalline Systems\thanks{
This work was funded by the National Key R \& D Program of China under grants 2019YFA0709600 and 2019YFA0709601.
The first author was also supported by the National Natural Science Foundation of China (No. 12301548) and the Open Project Program of Key Laboratory of Mathematics and Complex System (Grant No. K202302), Beijing Normal University.
The second author was also supported by the National Natural Science Foundation of China (No. 12371431).
}
}
\date{}

\author{Xiaoxu Li\thanks{
{\it xiaoxuli@bnu.edu.cn}.
Faculty of Arts and Sciences, Beijing Normal University, Zhuhai 519087, Guangdong, China.}
~and~ Huajie Chen\thanks{
Corresponding. {\it chen.huajie@bnu.edu.cn}.
School of Mathematical Sciences, Beijing Normal University, China.}
}

\begin{document}
\maketitle

\begin{abstract}
The multiple scattering theory (MST) is a Green's function method that has been widely used in electronic structure calculations for crystalline disordered systems.
The key property of the MST method is the scattering path matrix (SPM) that characterizes the Green’s function within a local solution representation. 
This paper studies various approximations of the SPM, under the condition that an appropriate reference is used for perturbation.
In particular, we justify the convergence of the SPM approximations with respect to the size of scattering region and the length of scattering path, which are the central numerical parameters to achieve a linear-scaling MST method.
We present numerical experiments on several typical systems to support the theory.
\end{abstract}

\section{Introduction}
\label{sec:introduction}

The disordered systems like random substitutional alloys play an important role in material sciences \cite{ebert11, ruban08} and pose a significant challenge in their electronic structure simulations. 
The difficulty arises from the inability to apply Bloch's theory \cite{bloch1929quantenmechanik} due to the absence of translational invariance.
There are primarily two approaches to perform electronic structure calculations for disordered systems.
One approach is the so-called supercell method \cite{braun2020sharp, cances2013mathematical}, which simulates a large cluster by replicating primitive cells with artificial boundary condition.
The supercell-based calculations are usually computationally demanding as one may need extremely large supercells to achieve the required accuracy.
Another category of approaches is the Green's function method, which effectively avoids the necessity to solve eigenvalue problems with artificial boundary conditions. 
One significant advantage of Green's function-based approaches is their inherent capability to describe the ``averaged" physical properties of disordered systems, establishing direct connections to physical observables \cite{ebert11, ruban08}. 
This attribute makes them particularly well-suited for simulations of alloy materials.
We mention that standard Green's function methods commonly refer to resolvent-based methods \cite{colbrook2023computing,li2018pexsi,lin16,thicke2021computing}, which still require discretizing the resolvent under some basis and then truncating it into algebraic linear systems.

As one of the most widely used Green's function methods, multiple scattering theory (MST), also known as the Korringa-Kohn-Rostoker (KKR) method \cite{kohn54, korringa47, zeller95}, is conceptually distinct from the framework of standard Green’s function methods. 
It provides a direct representation of the real-space Green's function without a Galerkin discretization procedure.
This approach offers a unified treatment for a wide range of material systems with defects, including substitutional atoms, interstitial impurities, surfaces and interfaces \cite{asato99, galanakis02, modinos00, nonas98}. 
The basic idea of MST is to partition the system into non-overlapping regions, solve each part separately and then combine the partial solutions to construct an overall Green's function directly by using the perturbation theory.
A distinctive feature of MST is the ability to completely separate the local potential aspects from the atomic configuration.
When a ``screened" potential is carefully chosen to construct the reference (see Section \ref{subsec:screened_kkr}), one could even achieve a ``linear-scaling" algorithm (see Section \ref{sec:linear-scaling}).
This scalability has been employed in numerous material simulations, especially in disordered systems for decades, see e.g. \cite{gorbatov21, mchugh20, saunderson20}.

There have been many efforts on improving the MST methods in practical simulations.
In \cite{zeller95} the authors introduced a screened form of MST that facilitates efficient evaluation of the Green's function.
Based on this formulation, \cite{thiess12, zeller08} proposed linear-scaling algorithms and demonstrated the convergence numerically.
In \cite{nicholson1994stationary, wang1995order}, a self-consistent MST method was proposed based on the observation that the electron density on a particular atom within a large system can be well approximated by considering only the electron scattering processes in a finite spatial region centered at that atom. 
We mention that there have also been efforts on exploiting the Green's function to study the topological insulators and edge states of extended systems, as highlighted in \cite{colbrook2021computing, colbrook2023computing, thicke2021computing} and the references therein.

Despite the success of MST in practice, the work on theoretical understanding and numerical analysis for this type of approach is very rare to our best knowledge.
There has been only one recent work \cite{li2023numerical} that provided an {\em a priori} error estimate for the MST method concerning the truncation of the angular momentum.
We point out that there has been mathematical works \cite{cances2008new, cances2008non, cances2013mean} devoted to studying the supercell method and justifying the convergence of the electronic structure models for disordered systems with respect to the supercell size.
We have not achieved a counterpart in this work, but will pursue it in our future work.
Finally, the MST method is highly related to a category of full-potential calculation methods, say the muffin-tin orbital (MTO) and the augmented plane wave (APW) methods, and we refer to some existing mathematical analysis on this aspect \cite{chen15a,chen2022augmented}.

The purpose of this work is to provide a rigorous analysis for the Green's function in MST.
We characterize the real-space Green's function through utilization of the scattering path matrix (SPM), which is fundamental in the MST formalism. 
Under the condition that an appropriate screened potential is chosen for the reference system, we establish the convergence of SPM with respect to the size of scattering region and the length of scattering path.

\vskip 0.2cm

{\bf Outline.} 
The rest of the paper is organized as follows.
In Section \ref{sec:preliminary}, we briefly discuss the MST method for disordered systems, and provides a representation of the Green's function by using SPM.
In Section \ref{sec:analysis}, we study numerical approximation of SPM and justify the convergence with respect to the size of scattering region and length of scattering path respectively.
In Section \ref{sec:numerics}, we present some numerical experiments to support the theory. 
Finally, we give some concluding remarks. 

\vskip 0.2cm

{\bf Notation.} 
In this paper, we will denote the cardinality of a finite discrete set $\mathcal{I}$ by $\#\mathcal{I}$.
The notations $\|\cdot\|_1$ and $\|\cdot\|_{\infty}$ are used to denote the $1$-norm and $\infty$-norm for a vector as well as the induced norms for a matrix.
The Schwartz space, the set of all rapidly decreasing smooth functions at positive infinity, will be denoted by $\Sc$ with
\begin{equation*}
\Sc = \Bigg\{f\in C^{\infty}\big(\R\big) :~ \sup_{x\geq 0}(1+|x|^2)^{t} \sum_{0\leq\alpha\leq\bar{\alpha}} |\partial^{\alpha}f(x)|<\infty \quad {\rm for~any}~t>0~{\rm and}~\bar{\alpha}\in \N \Bigg\} .
\end{equation*}
The symbol $C$ will denote a generic positive constant that may change from one line to the next, which will always remain independent of the approximation parameters under consideration and the choice of test functions.
The dependence of $C$ will normally be clear from the context or stated explicitly.

\section{Multiple scattering theory}
\label{sec:preliminary}
\setcounter{equation}{0} \setcounter{figure}{0}

The multiple scattering theory (MST) method formulates the real-space Green's function of a system by the scattering path matrix (SPM), from which the quantitative of interests can then be obtained.
In this section, we first briefly introduce a crystalline based disordered system,
then represent the Green's function by SPM within the MST framework, and finally discuss the choice of the screened potential for constructing an appropriate reference medium.

\subsection{Disordered systems}
\label{sec:disorder}

Let $d\in\{1,2,3\}$ be the space dimension.
We consider a $d$-dimensional disordered system embedded in a homogeneous lattice.
The homogeneous reference configuration is given by a Bravais lattice $\Lambda=A\Z^d$ with some non-singular matrix $A\in \mathbb{R}^{d \times d}$.
Then we can decompose the space $\R^d$ into non-overlapping cells $\{\Omega_n\}_{n\in\Lambda}$ based on the reference configuration $\Lambda$, such that $\Omega_n$ is centred at site $n\in\Lambda$ satisfying
\begin{align*}
\Omega_n\cup\Omega_k=\emptyset \quad\forall~n\neq k
\qquad{\rm and}\qquad
\R^d=\bigcup_{n\in\Lambda}\Omega_n .
\end{align*}
One of the most widely used decomposition is the Voronoi cells, defined by
\begin{eqnarray}\nonumber
\Omega_n := \Big\{\vr\in\R^d:~ |\vr-n|\le |\vr-k| \quad\forall~k\in\Lambda \Big\} \qquad{\rm for}~ n\in\Lambda .
\end{eqnarray}
We show in Figure \ref{fig:lattice} an example of the Voronoi cells for a two-dimensional triangular lattice.

\begin{figure}[!htb]
\centering
\includegraphics[width=16.0cm]{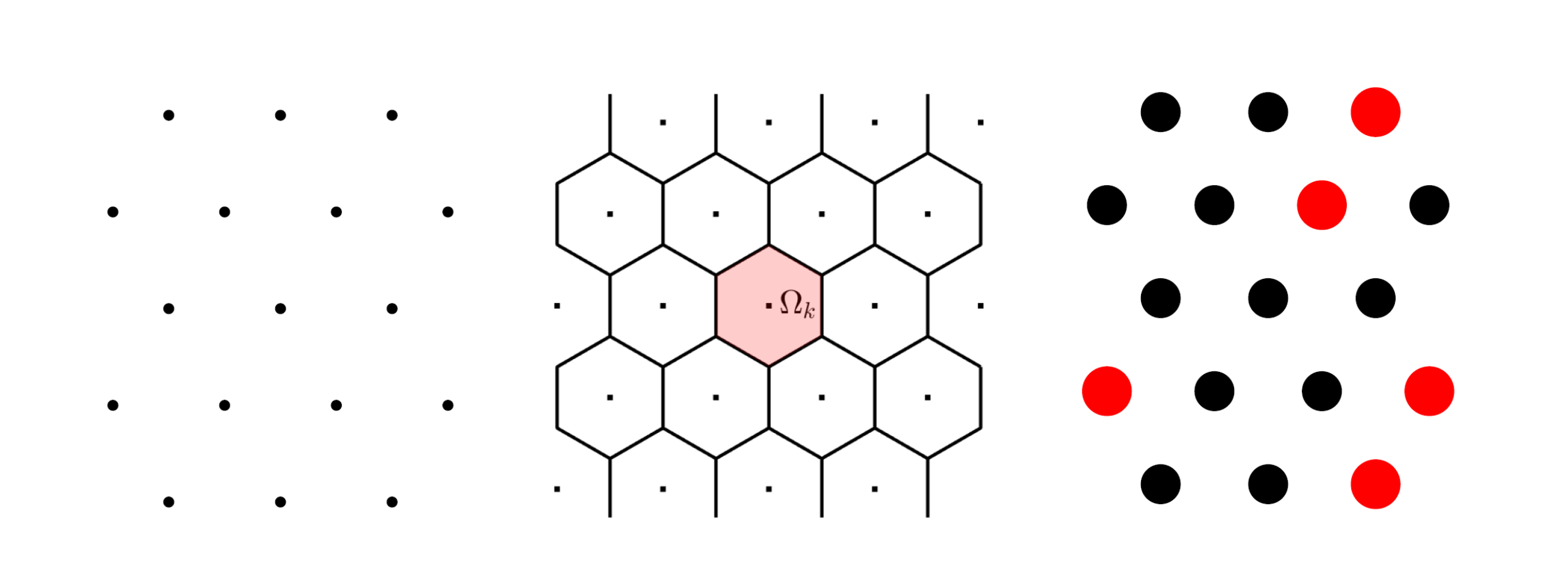}
\caption{Left: A triangular Bravais lattice. Middle: A space-filling decomposition with Voronoi cells. Right: A random two-component alloy configuration.}
\label{fig:lattice}
\end{figure}

The potential $V$ of the system is determined by the atomic species and positions in each cell $\Omega_n~(n\in\Lambda)$.
Let us denote the potential restricted on each cell by $v_n:=V|_{\Omega_n}$ for simplicity.
For a homogeneous lattice, the atomic configurations are the same in all cells, and hence $v_n\equiv v$ for any $n\in\Lambda$.
For a disordered system, the atomic configuration in each cell are random variables, which gives various $v_n$ in different cells.
Therefore, the system will lose the translation invariance.
We point out that the analysis in this work can be extended to more general crystalline systems with local defects (for example, impurities, vacancies, dislocations, etc.) without difficulty \cite{chen2019geometry, ehrlacher2016analysis}.

To study the electronic structure of a specific system with potential $V=\sum_{n\in\Lambda}v_n$, one needs to calculate the spectrum distribution of the Hamiltonian
\begin{eqnarray}
\label{ham}
\ham = -\Delta + V(\vr).
\end{eqnarray}
Note that a prefactor $1/2$ in front of the Laplacian is ignored throughout for simplicity of presentations.
Since the disordered atomic configuration destroys the translational symmetry of the system, the Bloch's theory can not be applied to obtain the ``bands" \cite{bloch1929quantenmechanik}.
A standard approach for such non-periodic systems is the so-called supercell method, which restricts the system in a bounded domain with an artificial periodic boundary condition \cite{braun2020sharp, cances2013mathematical}.
However, this method usually needs very big supercell to achieve the required accuracy, and hence requires huge computational cost.
Another widely used approach is the Green's function method, which focuses on simulating the Green's functions of systems so that one can effectively avoid solving the eigenvalues of the Hamiltonian.
This type of approach is particularly applicable to disordered systems.
The MST stands out as the most natural method for a direct representation of the real-space Green's function, which is viewed as the kernel of the resolvent, and resolves the spectral data of the Hamiltonian without a Galerkin discretization procedure \cite{economou06, gonis00}.

For a given energy parameter $z\in\C$, the real-space Green's function of $\ham$ can be defined by solution of the following differential equation with Delta function as the source term
\begin{eqnarray}
\label{model-green-eq}
\big( z-\ham\big) G({\vr}, {\vr}'; z) = \delta(\vr-\vr') \qquad \vr,\vr'\in\mathbb{R}^d
\end{eqnarray}
with the asymptotic boundary condition $G(\vr,\vr';z)\rightarrow 0$ as $|\vr-\vr'|\rightarrow \infty$.
It describes the propagation of an electron from point $\vr$ to $\vr'$ at energy $z$.
The real-space Green's function can provide a natural approach to characterize the ``local" density of states (DoS), and hence a localized description of the electronic properties of the system.

By using the spectrum theory, we can formulate the local DoS in cell $\Omega_n$ with a contour integration of the Green's function.
Let $f\in \Sc$ correspond to some physical observable, which admits an analytic continuation to the strip on the complex plane $\{z\in \C,~|{\rm Im} z|\leq \delta\}$.
Let $\mathscr{C}$ be a contour in the complex plane enclosing the spectrum of $\ham$ and avoiding all the singularities of $f$, and satisfy that for any $z\in\mathscr{C}$,
\begin{eqnarray}
\label{resolvant-dist}
\min\Bigg\{
{\rm dist}\big(z,\mathfrak{s}(f)\big), ~
{\rm dist}\left(z,\mathfrak{d}(\ham)\right)\Bigg\} \geq \frac{\delta}{2} ,
\end{eqnarray}
where $\mathfrak{d}(\ham)$ represents the spectrum set of $\ham$ and $\mathfrak{s}(f)$ represents the non-analytic region of $f$. 
We refer to left panel of Figure \ref{fig-contour} for an example of the Fermi-Dirac function (c.f. Remark \ref{remark:dos}), where the poles of $f$ are located on the imaginary axis and $\mathscr{C}$ is chosen to be a dumbbell-shaped contour. 
Note that the upper bound of the spectrum may not exist, and thus the contour could be left open (see right panel of Figure \ref{fig-contour} for example and \cite{levitt2020screening}), which is nevertheless not a concern since $f\in\Sc$ decays sufficiently fast. 
For simplicity of the presentation, we will only consider the closed contours in this work.
In addition, we also introduce an alternative Cauchy contour in \ref{append:contour}, which is more efficient for practical simulations.
Then the local DoS in cell $\Omega_n$ can be written as
\begin{eqnarray}
\label{ldos_n}
\mathscr{D}_n(f) := \frac{1}{2\pi i} \oint_{\mathscr{C}} f(z) \int_{\Omega_n} G(\vr,\vr;z) \dd\vr \dd z
\qquad n\in\Lambda.
\end{eqnarray}

\begin{figure}[htbp]
\centering
\includegraphics[width=8.cm]{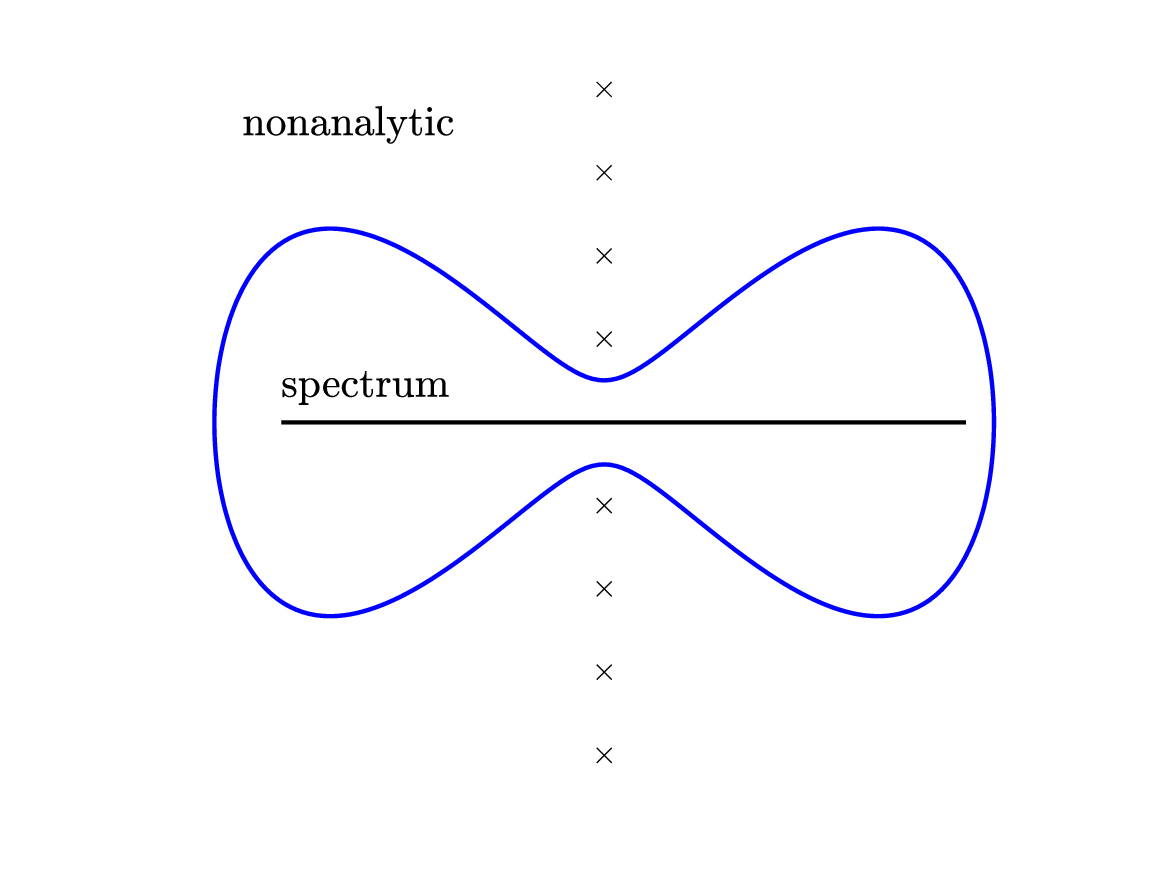}
\includegraphics[width=8.cm]{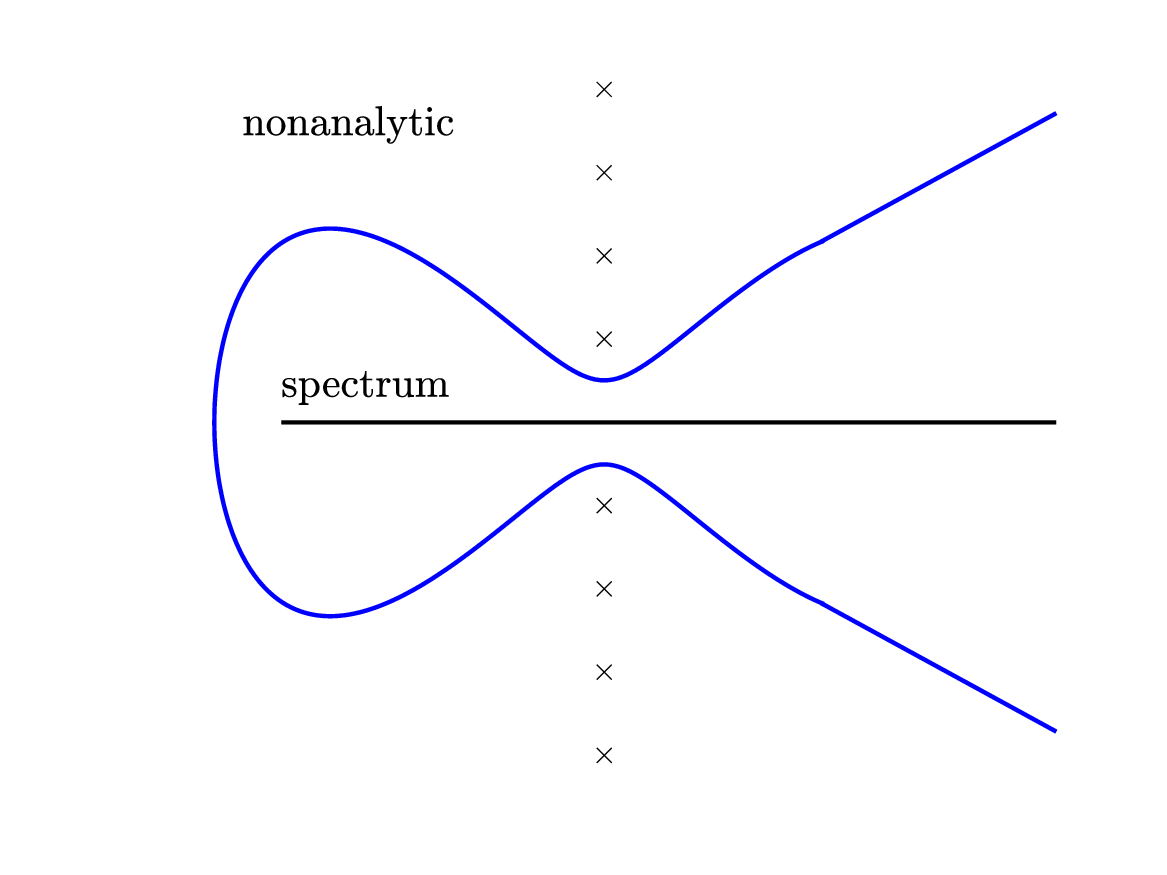}
\vskip -0.5cm
\caption{A schematic illustration of a Cauchy contour. Left: a dumbbell-shaped contour. Right: an open contour.}
\label{fig-contour}
\end{figure}

\begin{remark}
[Electron density as DoS with Fermi-Dirac function]
\label{remark:dos}
A commonly used observable related to the particle number is the Fermi-Dirac function
\begin{eqnarray}
\label{fermi-dirac-eq}
f_{\rm FD}(z) = 1/\big(1+e^{(z-\mu)/(k_B T)}\big)
\end{eqnarray}
with the chemical potential $\mu$, the Boltzmann constant $k_B$, and the electron temperature $T$.
The electron density can then be formulated by Green's function \cite{economou06, zeller08}
\begin{eqnarray}
\label{density-eq}
\rho(\vr) = \frac{1}{2\pi i}\oint_{\mathscr{C}} f_{\rm FD}(z)  G(\vr,\vr;z) \dd z
\qquad\vr\in\R^d,
\end{eqnarray}
and $\int_{\Omega_n}\rho(\vr) \dd\vr = \mathscr{D}_n(f_{\rm FD})$ denotes the number of states in $\Omega_n$.
The Fermi-Dirac function is analytic except for the poles at 
\begin{eqnarray}
\label{matsubara-def}
z_j = \mu + i\pi \big(2j-{\rm sgn}(j)\big) k_B T \qquad j\in\Z\backslash\{0\},
\end{eqnarray}
which are referred to as {\it Matsubara frequencies} in condensed-state physics \cite{wildberger1995fermi}.
\end{remark}

\begin{remark}
[``Averaged" DoS for disorded systems]
For a finite system, the ``total" DoS of the system is directly given by summing the local DoS over all cells $\Omega_n~(n \in \Lambda)$.
However, for an extended system the lattice $\Lambda$ is infinite, and a careful ``average" of the local DoS is required.
This is a difficult problem for generic extended systems, and we refer to \cite{cances2011local,cances2013mathematical,chen2018thermodynamic,chen16,ehrlacher2016analysis} for systems with local defects and \cite{cances2013mean} for disordered systems. 
\end{remark}

\subsection{Green's function by the scattering path matrix}
\label{sec:Green_MSY}

The key feature of the MST Green's function formulation is the separation of the two aspects of the problem: the scattering (chemical identity of the atoms) and the structure (positions of the atoms).
These two aspects are characterized by the ``cell scattering matrix" $t$ and the ``structure constants" $g$, respectively.
We will use these two aspects to define the scattering path matrix $\tau$, and then formulate the Green's function of the system.
These properties are usually difficult to handle computationally in the coordinate representation.
Instead, they are often transferred into matrices whose elements are labeled by indices $(\ell, m)$ of the angular momentum eigenstates.

First, the ``cell scattering matrix" represents the scattering arising from a specific cell $\Omega_n$.
Let $v_n:=V|_{\Omega_n}~(n\in\Lambda)$ be the potential restricted on each cell.
Let $\zeta_{n,\ell m}(\vr;z)$ and $\xi_{n,\ell m}(\vr;z)$ be the regular and irregular solutions, respectively, which satisfy the equation $(-\Delta + v_n + z) \hskip 0.1cm u = 0$ in $\Omega_n$ with the angular quantum number $\ell\in\N$ and $|m|\leq \ell$, but with different boundary conditions on $\partial\Omega_n$.
We refer to \ref{append:single-site} for explicit expressions of the local solution $\zeta_{n,\ell m}$ and $\xi_{n,\ell m}$ for spherical potentials (see also \cite[Appendix C]{li2023numerical} for non-spherical potentials).
Then we can define the cell scattering
\begin{eqnarray}
\label{t:cell}
t_{n,\ell m}(z) := \int_{\Omega_n} J_{\ell m}(\vr-n;z) v_n(\vr) \zeta_{n,\ell m}(\vr;z) \dd\vr 
\qquad \ell\in\N,~|m|\leq\ell ,
\end{eqnarray}
where $J_{\ell m}(\vr;z):=j_{\ell}\big(\sqrt{z}|\vr|\big) \ylm$ with $j_{\ell}$ the spherical Bessel functions and $\hat{\vr}=\vr/|\vr|$.
Note that $J_{\ell m}(\vr;z)$ solves the Helmholtz equation $(-\Delta + z) J_{\ell m}(\cdot;z)= 0$ and can be viewed as a free-electron incoming wave.
Then $t_{n,\ell m}(z)$ determines the relationship between the incoming wave and outgoing wave (scattered by the potential $v_n$ in cell $\Omega_n$) at energy $z$ with angular component $Y_{\ell m}$.
The above quantities can all be obtained ``locally" in parallel on each cell $\Omega_n$, for which the computational complexity is relatively small in the MST framework. 

Then, the ``structure constant" $g$ depends on the lattice structure $\Lambda$ and energy parameter $z$. 
For $n,n'\in\Lambda$ and $n\ne n'$, define
\begin{eqnarray}
\label{structure-constant}
g_{nn',\ell m\ell'm'}(z) := 4\pi \sum_{\ell '' m''} i^{l-l'-l''} C_{\ell m, \ell' m', \ell''m''} H_{\ell'' m''}(n'-n;z)
\end{eqnarray}
where $\displaystyle C_{\ell m, \ell' m', \ell''m''} = \int_{S^2} Y_{\ell m}(\hat{\bm{r}}) Y_{\ell' m'}(\hat{\bm{r}}) Y_{\ell''m''}(\hat{\bm{r}}) \dd\hat{\bm{r}}$ are the Gaunt coefficients (see, e.g., \cite{mavropoulos06}), and $H_{\ell m}(\vr;z): = h_{\ell}\big(\sqrt{z}|\vr|\big) \ylm$ with $h_{\ell}$ the spherical Hankel functions.
By convention, the structure constant is set to be $0$ for $n=n'$.
In particular, the structure constants $g$ gives the expansion coefficients of the Hankel functions $H_{\ell m}$ centered at site $n$ with respect to the Bessel functions $J_{\ell m}$ centered at $n'$ \cite{li2023numerical,martin05}.

Now we are able to define the key property in MST: the scattering path matrix (SPM), also known as the structural Green's function in physic literature.
Let $g(z):=\big\{g_{nn',\ell m\ell' m'}(z)\big\}_{nn',\ell m\ell' m'}$ and $t(z):=\big\{t_{n,\ell m}(z)\delta_{nn'}\delta_{\ell\ell'}\delta_{mm'}\big\}_{nn',\ell m\ell' m'}$. 
Then the SPM is defined as
\begin{eqnarray}
\label{spm}
\tau(z) := \Big( I - g(z) \hskip 0.1cm t(z) \Big)^{-1} \hskip 0.1cm g(z) 
\qquad z\in\C ,
\end{eqnarray}
where $I$ is the identity operator.
This is well defined when $I-g(z)t(z)$ is invertable, and the singularities appear when $z$ belongs to the spectrum of the Hamiltonian, i.e. when $z\in \mathfrak{d}(\ham)$ \cite{li2023numerical}.
The abstract matrix form \eqref{spm} can be rewritten in the following equation form
\begin{eqnarray}
\label{spm:alt}
\tau(z) = g(z) + g(z) \hskip 0.1cm t(z) \hskip 0.1cm \tau(z)
\end{eqnarray}
with an explicit expression for each matrix element
\begin{eqnarray}
\label{spm:matrix}
\tau_{nn',\ell\ell'mm'}(z) = g_{nn',\ell m\ell'm'}(z) + \sum_{n''\in\Lambda}\sum_{\ell''m''} g_{nn'',\ell m\ell''m''}(z) t_{n'',\ell''m''}(z) \tau_{n''n',\ell''m''\ell'm'}(z) .
\end{eqnarray}

In the MST formalism, the Green's function can be written in terms of the SPM \cite{gonis00, gonis1989multiple, zeller08, zeller95}.
For $z\in\C \backslash \mathfrak{d}(\ham)$ and $\vr\in\Omega_n,~\vr'\in\Omega_{n'}$, it is given by
\begin{align}
\label{green:angular:tau}
\nonumber
G(\vr,\vr';z) 
& = -i\delta_{nn'} \sqrt{z} \sum_{\ell m} \zeta_{n,\ell m}(\vr_{n<};z) \xi_{n,\ell m}(\vr_{n>};z) 
\\
& \quad\quad + \sum_{\ell m} \sum_{\ell' m'} \zeta_{n,\ell m}(\vr;z) \tau_{nn',\ell m\ell'm'}(z) \zeta_{n',\ell' m'}(\vr';z) 
\end{align}
with $\vr_{n>}=\arg\max\big\{|\vr-n|, |\vr'-n|\big\}$ and $\vr_{n<}=\arg\min\big\{|\vr-n|, |\vr'-n|\big\}$.
The first part of \eqref{green:angular:tau} is also referred to single-site Green's function.
With the MST expression \eqref{green:angular:tau} for Green's function and \eqref{ldos_n} for the local DoS, the local properties of a disordered system can be obtained once we have the SPM.
This Green's function formulation separates atomic potential and structure of the systems, since \eqref{spm} can be written by
$\tau(z) = \big( g(z)^{-1} - t(z) \big)^{-1}$ with $t(z)$ incorporating all the effects of the potential inside each cell $\Omega_n~(n\in\Lambda)$ and $g(z)$ depending only upon the structure of the lattice $\Lambda$. 

As the above resulting formalism, the MST method is usually given in an angular momentum representation, which has been found of great value in practical numerical applications.
However, in this work we will focus on the effects of spatial truncation of the scattering path, and the angular momentum indices $(\ell,m)$ do not essentially bring any insight to our analysis.
Therefore we will ignore the angular momentum indices in the rest of this paper for simplicity of presentations, and refer to our recent work \cite{li2023numerical} for a careful analysis on the angular momentum cutoff.
In particular, \eqref{green:angular:tau} can be rewritten as 
\begin{align}
\label{green:tau}
G(\vr,\vr';z) = -i\delta_{nn'} \sqrt{z} \zeta_{n}(\vr_{n<};z) \xi_{n}(\vr_{n>};z) 
+ \zeta_{n}(\vr;z) \tau_{nn'}(z) \zeta_{n'}(\vr';z) , 
\end{align}
where we have slightly abused the notation by writing the blocks $\zeta_{n}(\vr;z)=\big\{\zeta_{n,\ell m}(\vr;z)\big\}_{\ell m}$, $\xi_{n}(\vr;z)=\big\{\xi_{n,\ell m}(\vr;z)\big\}_{\ell m}$ and $\tau_{nn'}(z)=\big\{\tau_{nn',\ell m\ell'm'}(z)\big\}_{\ell m\ell'm'}$ for $n, n'\in\Lambda$ without angular momentum indices.
Similarly, we will also use the abbreviations $t_n(z)$ for \eqref{t:cell} and $g_{nn'}(z)$ for \eqref{structure-constant}.
From a physical perspective, each block $\tau_{nn'}(z)$ describes the propagation of the electron from site $n$ to $n'$ by traversing all possible paths across the lattice sites over the whole space.

By using \eqref{ldos_n} and \eqref{green:tau}, we see that to evaluate the local DoS in cell $\Omega_n$, one only needs the real-space Green's function for $\vr=\vr'\in\Omega_n$ and hence the diagonal element $\tau_{nn}$ of SPM.
Therefore, we will focus on numerical algorithms for the diagonal blocks of SPM, more precisely, by taking $n=n'$ on the left-hand side of \eqref{spm:matrix}.

\subsection{The perturbation perspective and choice of reference system}
\label{subsec:screened_kkr}

The MST Green's function formalism presented above can also be derived from the perturbation theory \cite{economou06, gonis00}, where the unperturbed system is chosen as free-electron gas.
Note that the Green's function $G^0$ of a free-electron gas (for $d=3$) has the analytic form 
$$
G^0(\vr,\vr';z) = - \frac{e^{i\sqrt{z}|\vr-\vr'|}}{4\pi |\vr-\vr'|},
$$
whose corresponding SPM is nothing but the structure constant $g(z)$ given in \eqref{structure-constant} \cite{gonis00, mavropoulos06}. 
From the viewpoint of perturbation theory, the Green's function of a perturbed system can be expressed in terms of the Green’s function of the reference system by using the so-called Dyson equation.
Let $G^{\rm r}$ be the Green's function of a reference system with given potential $V^{\rm r}$. Then the Green's function $G$ of a perturbed system with potential $V$ can be obtained by the Dyson equation \cite{gonis00, mavropoulos06}
\begin{eqnarray*}
G(\vr,\vr';z) = G^{\rm r}(\vr,\vr';z) + \int_{\R^d} G^{\rm r}(\vr,\vr'';z) \bigg(V(\vr'')-V^{\rm r}(\vr'')\bigg) G(\vr'',\vr';z) \dd\vr'' .
\end{eqnarray*} 
The SPM definition \eqref{spm} can be viewed as the corresponding Dyson formulation for SPM \cite{zeller08, zeller95}, by taking the free-electron gas (i.e. $V^{\rm r}(\vr)\equiv 0$) as the reference system.
We will show in the following that the SPM can be derived from the perturbation method by taking more appropriate reference systems.

Let $t^{\rm r}(z)$ be the cell scattering matrx with a reference potential $v_n^{\rm r}$ in \eqref{t:cell}.
Note that \eqref{spm} implies 
\begin{eqnarray*}
\tau(z)^{-1} - \tau^{\rm r}(z)^{-1} = t^{\rm r}(z) - t(z) =: -\Delta t(z).
\end{eqnarray*}
We have that $\Delta t(z)$ is block diagonal, whose diagonal part can be explicitly calculated with
\begin{eqnarray}
\label{t:cell:perturbation}
\Delta t_{n}(z) = \int_{\Omega_n} \zeta^{\rm r}_{n}(\vr;z) \big( v_n(\vr) - v^{\rm r}_n(\vr) \big) \zeta_{n}(\vr;z) \dd\vr 
\qquad n\in\Lambda.
\end{eqnarray}
This equality allows us to express $\tau(z)$ as a perturbed SPM through its counterpart $\tau^{\rm r}(z)$ of the unperturbed system
\begin{eqnarray}
\label{spm:ref}
\tau(z) 
= \Big(I-\tau^{\rm r}(z) \hskip 0.1cm \Delta t(z) \Big)^{-1} \tau^{\rm r}(z) .
\end{eqnarray}
This can alternatively be written by an implicit equation form
\begin{eqnarray}
\label{spm:dyson}
\tau(z) = \tau^{\rm r}(z) + \tau^{\rm r}(z) \hskip 0.1cm \Delta t(z) \hskip 0.1cm \tau(z) .
\end{eqnarray}

For the free-electron gas, we have the reference SPM $\tau^{\rm r}(z) = g(z)$ and the reference potential $v_n^{\rm r}(\vr)\equiv 0$ for all $n\in\Lambda$.
In practical calculations, more ``appropriate" reference system can be chosen for \eqref{spm:dyson} instead of the free-electron gas.
The reference potential is naturally taken to be periodic with respect to the lattice $\Lambda$, so that the reference Green's function $G^{\rm r}$ and reference SPM $\tau^{\rm r}$ are easily obtained.
In addition, we make the following assumptions 
on the choice of reference.
First, we assume that the reference SPM has a sufficiently fast off-diagonal decay.

\vskip 0.2cm 

\noindent
{\bf (A1)}
There exist constants $C^{\rm r}>0$ and $\gamma^{\rm r}>0$, such that for any $z\in\C$ satisfying \eqref{resolvant-dist},
\begin{equation}
\label{tau-assumption}
\big|\tau^{\rm r}_{jk}(z)\big| \le C^{\rm r} e^{-\gamma^{\rm r}|j-k|} \qquad \forall~j,k\in\Lambda .
\end{equation}

\vskip 0.2cm 

We will provide some intuitively discussion on why assumption {\bf (A1)} makes the criteria of a ``good" reference system.
The bottleneck of calculating the SPM is to perform the operator inversion in \eqref{spm:ref}.
If the matrix elements have sufficiently fast off-diagonal decay, then one could exploit various techniques to reduce the computational cost of matrix inversion \cite{thiess12,wang1995order,zeller08} (see also our numerical approaches in Section \ref{sec:analysis}).
Note that $\Delta t(z)$ is block diagonal, hence the off-diagonal decay of $I-\tau^{\rm r}(z)\Delta t(z)$ is determined by the reference SPM $\tau^{\rm r}(z)$.
Therefore, the exponential off-diagonal decay $\tau^{\rm r}(z)$ in \eqref{tau-assumption} could allow for efficient calculations of the SPM with \eqref{spm:ref}.
We further refer to the analysis in Section \ref{sec:analysis} and numerical experiments in Section \ref{sec:numerics} for more justifications, and also the literature \cite{mavropoulos06, zeller08, zeller95}

\begin{remark}[Principle of choosing the reference]
\label{remark:tau-decay}
For the energy parameter $z\in\C$ satisfying \eqref{resolvant-dist}, the off-diagonal decay of $\tau^{\rm r}(z)$ corresponds to the decay of the reference Green's function $G^{\rm r}(\vr,\vr',z)$ with respect to $|\vr-\vr'|$.
This is determined by the distance between $z$ and the spectrum of reference Hamiltonian $-\Delta + V^{\rm r}$ \cite{baer1997sparsity,benzi13,e11,lin16}.
If the reference system is not carefully chosen, say the free-electron gas, then the singularities $\mathfrak{s}(f)$ could force $z$ to get close to the spectrum of the reference spectrum and lead to a slow off-diagonal decay of the reference SPM.
We show in Figure \ref{fig-contour} an example of the Fermi-Dirac function \eqref{fermi-dirac-eq}, which has two singularities (i.e., the Matsubara points $z_1$ and $z_{-1}$ defined in \eqref{matsubara-def}) that are very close to the real axis at low temperature.
Ideally, one should find a reference system such that it can enhance the separation between the spectrum of reference Hamiltonian and the contour $\mathscr{C}$ satisfying \eqref{resolvant-dist}.
\end{remark}

In the following, we will provide an example that is widely used in MST community, to show that the fast off-diagonal decay condition can be fulfilled by carefully choosing a reference potential.
Let us consider an one-dimensional system, and $f\in\Sc$ to be the Fermi-Dirac function.
We take the following rectangular potential as a reference system (see the left panel of Figure \ref{fig-ref-potential} for a schematic plot)
\begin{eqnarray}
\label{ref-potential-def}
V^{\rm r}(\vr) = \left\{
\begin{array}{ll}
H \qquad & |\vr-n|\le W ~{\rm for}~n\in\Lambda, \\[1ex]
0 \qquad & {\rm otherwise}. 
\end{array}\right.
\end{eqnarray}
The potential is a constant repulsive potential within the ball centered at each site in $\Lambda$ and vanishes in the interstitial region.
We can conveniently tune this simple potential by its two parameters (the width $W$ and height $H$), such that the bottom of the reference spectrum can be raised. 
Comparing with the free-electron gas, the spectrum of this rectangular potential system is significantly further away from the singularities of $f$, and hence away from the contour in \eqref{ldos_n} (see the right panel of Figure \ref{fig-ref-potential}). 
This feature guarantees that the elements of reference SPM $\tau^{\rm r}(z)$ decay exponentially fast away from the diagonal along the whole contour.

\begin{figure}[htbp]			
\centering
\includegraphics[width=8cm]{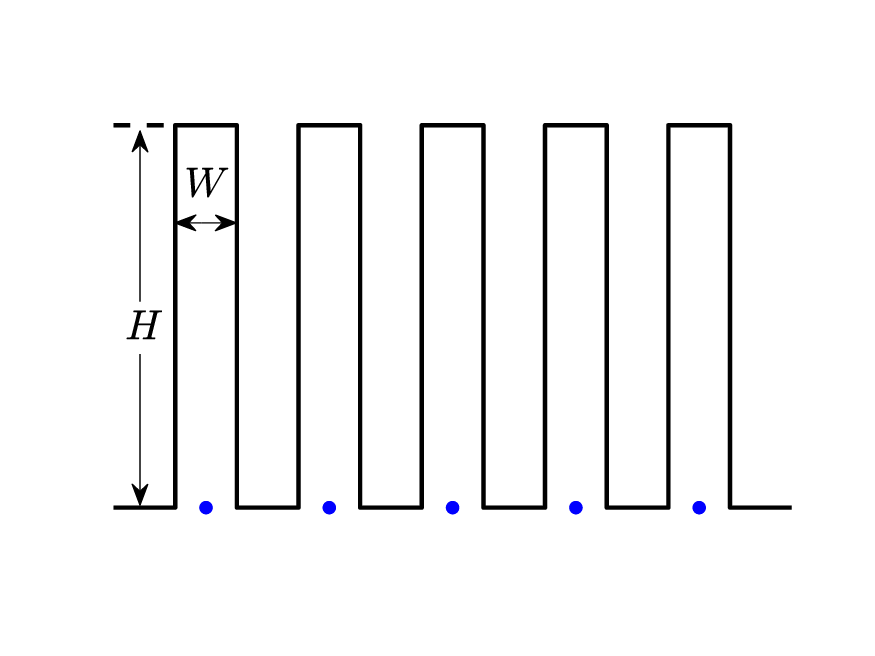}
\includegraphics[width=8cm]{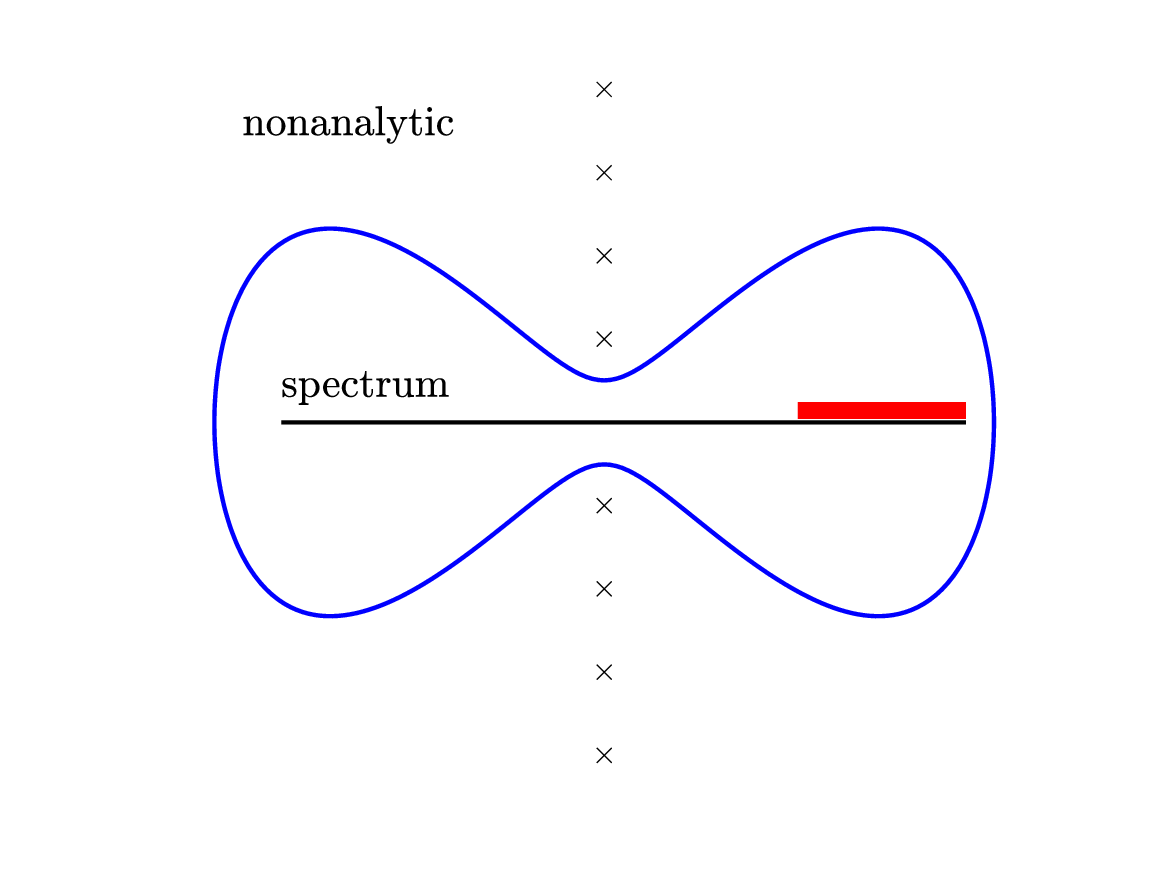}
\caption{Left: An one-dimensional schematic illustration of rectangular reference potential. 
Right: A schematic illustration of Cauchy contour, where the black and red thick lines represent the spectrum of $\mathcal{H}$ and $\mathcal{H}^{\rm r}$, respectively.}
\label{fig-ref-potential}
\end{figure}

Besides the condition on fast off-diagonal decay of the reference SPM, we will also need the following assumption on the reference to make the perturbation method \eqref{spm:ref} is well behaved.
For a subdomain $\Pi\subset\Lambda$, we introduce the notations $\tau^{{\rm r},\Pi}(z):=\big\{\tau_{nn'}^{\rm r}(z)\big\}_{n,n'\in\Pi}$ and $\Delta t^{\Pi}(z):=\big\{\Delta t_{nn'}(z)\big\}_{n,n'\in\Pi}$, that represent the restrictions of $\tau^{\rm r}(z)$ and $\Delta t(z)$ on $\Pi$ respectively.

\vskip 0.2cm

\noindent
{\bf (A2)}
For any $z\in\C$ satisfying \eqref{resolvant-dist}, the operator $I-\tau^{\rm r}(z) \hskip 0.1cm \Delta t(z)$ is invertible.
Moreover, $I-\tau^{{\rm r},\Pi}(z) \hskip 0.1cm \Delta t^{\Pi}(z)$ is invertible for any $\Pi\subset\Lambda$, and there exists a constant $\sigma>0$ such that
\begin{eqnarray}
\Big\| \Big( I-\tau^{{\rm r},\Pi}(z) \hskip 0.1cm \Delta t^{\Pi}(z) \Big)^{-1} \Big\|_2 \le \sigma 
\qquad {\rm for~any}~\Pi\subset\Lambda ~{\rm and}~ z\in\C~{\rm satisfying}~\eqref{resolvant-dist} .
\end{eqnarray} 

\vskip 0.2cm

The condition {\bf (A2)} assumes that the contour $\mathscr{C}$ used in \eqref{ldos_n} does not intersect with the spectrum of any of the subsystems under the reference background.
%
In practical calculations by the MST methods, one commonly implicitly assumes the norm of $\tau^{{\rm r}}(z)\Delta t(z)$ could be small \cite{thiess12, zeller08, zeller95} instead of {\bf (A2)}.
We write the condition in the following, which can be viewed as a stronger condition than {\bf (A2)}.

\vskip 0.2cm 

\noindent
{\bf (A3)}
There exists a constant $\kappa\in(0,1)$ such that for any $z\in\C$ satisfying \eqref{resolvant-dist}
\begin{equation}
\label{spectrum-assumption2}
\big\|\tau^{{\rm r}}(z)\Delta t(z)\big\|_{1} \leq \kappa . 
\end{equation}

Given the fast off-diagonal decay of $\tau^{{\rm r}}(z)$ in {\bf (A1)}, the condition {\bf (A3)} indicates that the perturbation on the reference system is relatively small.
This can imply a strictly column diagonal dominance of $I-\tau^{\rm r}(z) \hskip 0.1cm \Delta t(z)$, and then give rise to the condition {\bf (A2)} from the Ger\v{s}gorin circle theorem.

Based on a Combes-Thomas type estimate \cite{combes1973asymptotic}, we can provide the following lemma to show the exponential off-diagonal decay of $\Big( I-\tau^{{\rm r},\Pi}(z) \hskip 0.1cm \Delta t^{\Pi}(z) \Big)^{-1}$ for any $\Pi\subset\Lambda$, whose proof is similar to that in \cite[Lemma 6]{chen16}.
This argument is crucial to the convergence of numerical approximation for SPM in the next section.
\begin{lemma}
\label{lem:exp-decay}
Assume that the reference system satisfies the conditions {\bf (A1)} and {\bf (A2)}. 
Then there exist positive constants $\gamma$ and $C$ depending on $\sigma$ and $\gamma^{\rm r}$, such that for any $\Pi\subset\Lambda$ and $z\in \C$ satisfying \eqref{resolvant-dist},
\begin{eqnarray}
\label{inv-X-est}
\left| \Big( I-\tau^{{\rm r},\Pi}(z) \hskip 0.1cm \Delta t^{\Pi}(z) \Big)_{jk}^{-1} \right| \le Ce^{-\gamma |j-k|} 
\qquad \forall~j,k\in \Pi .
\end{eqnarray}
\end{lemma}

\section{Numerical approximations of SPM}
\label{sec:analysis}
\setcounter{equation}{0} \setcounter{figure}{0}

In this section, we will construct two efficient algorithms for the SPM approximations, and derive the error estimates with respect to the corresponding numerical parameters.

\subsection{Truncation of scattering region}
\label{sec:ldos}
 
As \eqref{spm:dyson} (elaborated by \eqref{spm:matrix}) involves an infinite sum of scattering sites over the lattice $\Lambda$, we will first introduce a truncation of the scattering region for SPM calculations.
In particular, \eqref{spm:dyson} can be written in an explicit matrix element form (similar to \eqref{spm:matrix}, but without the angular momentum indices)
\begin{eqnarray}
\label{spm:elment:ref}
\tau_{nn'}(z) = \tau^{\rm r}_{nn'}(z) + \sum_{n''\in\Lambda} \tau^{\rm r}_{nn''}(z) \hskip 0.1cm \Delta t_{n''}(z) \hskip 0.1cm \tau_{n''n'}(z) ,
\end{eqnarray}
in which the infinite sum over $\Lambda$ should be replaced by some finite subset in practice.
Note that to evaluate the local DoS \eqref{ldos_n} for cell $\Omega_n$, we only need to consider the diagonal block $\tau_{nn}$.

Let $R>0$ denote the size of scattering region.
For a given site $n\in\Lambda$, we define the following finite scattering region
\begin{eqnarray*}
\Lambda^{\dist}_{n} := \big\{k\in\Lambda: ~|k-n|\le \dist \big\} \subset \Lambda .
\end{eqnarray*}
Then we restrict the site indices in \eqref{spm:elment:ref} to this finite region, that is for any $n',n''\in\Lambda_n^R$, the SPM is approximated by
\begin{eqnarray}
\label{green-matrix-dyson-eq}
\tau^R_{n'n''}(z) = \tau^{{\rm r},\dist}_{n'n''}(z) + \sum_{n'''\in\Lambda_n^R} \tau^{{\rm r},\dist}_{n'n'''}(z) \hskip 0.1cm \Delta t_{n'''}(z) \hskip 0.1cm \tau^R_{n'''n''}(z).
\end{eqnarray}
Note that the scattering region $\Lambda^{\dist}_{n}$ as well as approximate SPM $\tau^R$ is designed specifically to evaluate the local DoS for cell $\Omega_n$, rather than other cells in $\Lambda^{\dist}_{n}$.
The dependency of approximate SPM $\tau^R$ on the central site $n$ is suppressed for simplicity of notations.

This truncation of the scattering region can be written in the matrix form
\begin{eqnarray}
\label{SPM:R_implicit}
\tau^{R}(z) = \tau^{{\rm r},R}(z) + \tau^{{\rm r},R}(z) \hskip 0.1cm \Delta t^{R}(z) \hskip 0.1cm \tau^{R}(z),
\end{eqnarray}
where $\Delta t^{\dist}(z):=\big\{\Delta t_{n}(z)\delta_{nn'}\big\}_{n,n'\in\Lambda^R_n}$ and $\tau^{{\rm r},R}(z):=\big\{\tau^{{\rm r},\dist}_{nn'}(z)\big\}_{n,n'\in\Lambda^R_n}$ denote the restriction of the corresponding operators in $\Lambda^{\dist}_{n}$.
Note that \eqref{SPM:R_implicit} gives an implicit equation as the SPM $\tau^R(z)$ appears on both the left- and right-hand sides of \eqref{SPM:R_implicit}.
In practical calculations, the truncated SPM can be obtained by performing the following matrix inversion
\begin{eqnarray}
\label{SPM:R}
\tau^{R}(z) = \Big(I - \tau^{{\rm r},R}(z) \hskip 0.1cm \Delta t^R(z)\Big)^{-1} \hskip 0.1cm \tau^{{\rm r},R}(z),
\end{eqnarray}
where the existence of matrix inversion is guaranteed by {\bf (A2)}.

In general, the matrix inversion amounts to solving the linear system 
\begin{eqnarray}   
\label{linear-system}
\Big( I - \tau^{{\rm r},\dist}(z) \hskip 0.1cm \Delta t^R(z) \Big) \tau^{\dist}(z) =  \tau^{{\rm r},\dist}(z),
\end{eqnarray}
which requires $\mathcal{O}(R^{3d})$ operations by a direct solution for each energy parameter $z\in\mathscr{C}$.
We point out that the computation of SPM in a bounded scattering region has already achieved linear scale concerning the size of simulation system, since the computational cost of local DoS only depends on the size of scattering region $\Lambda^{\dist}_{n}$ but not on size of the total system.

In addition, note that the coefficient matrix of \eqref{linear-system} is a non-Hermitian matrix with exponential off-diagonal decay.
There are various existing methods for solving linear systems with a sparse non-Hermitian matrix, such as generalized minimal residual method (GMRES) \cite{saad1986gmres}, biconjugate gradient method (BCG) \cite{fletcher2006conjugate}, quasi-minimal residual method (QMR) \cite{freund1991qmr}, and transpose-free quasi-minimal residual method (TFQMR) \cite{freund1993transpose}.
%

We show in the following theorem that the approximations of the diagonal block corresponding to the central cell $n$ of the SPM can converge exponentially fast with respect to the cutoff size $\dist$.

\begin{theorem}
\label{thm:spm-exp-convergence}
Assume that the reference system satisfies the conditions {\bf (A1)} and {\bf (A2)}. 
Then for any $n\in\Lambda$ and $z\in\C$ satisfying \eqref{resolvant-dist}, there exist positive constants $C$ and $\eta$ independent on $z$ and $\dist$, such that
\begin{eqnarray}
\label{tau-site-truncate-err}
\big|\tau^{\dist}_{nn}(z) - \tau_{nn}(z)\big| \le C e^{-\eta\dist} .
\end{eqnarray}
\end{theorem}

\begin{proof}
Denote $\displaystyle M^{\dist}(z)=I-\tau^{{\rm r},\dist}(z) \hskip 0.1cm \Delta t^{\dist}(z)$ for $R>0$.
For $n\in\Lambda$ and $\dist_2>\dist_1$, let $N_1=\# \Lambda_n^{\dist_1}=\mathcal{O}(R_1^d)$ and $N_2=\# \Lambda_n^{\dist_2}=\mathcal{O}(R_2^d)$.
Since $\Delta t^{\dist}(z)$ is block diagonal, we can divide $M^{\dist_2}(z)$ into the following form with the appropriate relabeling
\begin{eqnarray}\nonumber
M^{\dist_2}(z) = :\left(
\begin{matrix}
M^{\dist_1}(z) & B(z) \\
D(z) & E(z)  \\
\end{matrix}
\right)
\end{eqnarray}
with $B(z)\in\mathbb{C}^{N_1\times(N_2-N_1)}$, $D(z)\in\mathbb{C}^{(N_2-N_1)\times N_1}$ and $E(z)\in\mathbb{C}^{(N_2-N_1)\times(N_2-N_1)}$.
Then we derive that
\begin{align}\nonumber
&\Big(M^{\dist_2}(z)\Big)^{-1} - \left(
\begin{matrix}
\Big(M^{\dist_1}(z)\Big)^{-1} & O \\
O & O  \\
\end{matrix} 
\right) \\[1ex]\nonumber
= & 
\left(
\begin{matrix}
M^{\dist_1}(z) & B(z) \\
D(z) & E(z)  \\
\end{matrix}
\right)^{-1} - 
\left(
\begin{matrix}
\Big(M^{\dist_1}(z)\Big)^{-1} & O \\
O & O  \\
\end{matrix}
\right)	\\[1ex]\nonumber
= & \left(
\begin{matrix}
\Big(M^{\dist_1}(z)-B(z) E^{-1}(z) D(z)\Big)^{-1} - \Big(M^{\dist_1}(z)\Big)^{-1} & * \\
* & *  \\
\end{matrix}
\right) \\[1ex]\nonumber
= & 	
\left(
\begin{matrix}
\Big(M^{\dist_1}(z)-B(z) E^{-1}(z) D(z)\Big)^{-1} B(z) E^{-1}(z) D(z)\Big(M^{\dist_1}(z)\Big)^{-1} & * \\
* & *  \\
\end{matrix}
\right).
\end{align}
Let $X^R(z)=\Big(M^R(z)\Big)^{-1}$. 
For $\alpha,\beta\in\Lambda^{R_1}_n$, it gives rise to 
\begin{eqnarray}
\label{inv-err-eq-1}
|X^{\dist_2}_{\alpha\beta}(z)-X^{\dist_1}_{\alpha\beta}(z)| \le \sum_{j,k\in\Lambda_n^{R_1}} |X^{\dist_2}_{\alpha j}(z)| \cdot \left|\Big(B(z) E^{-1} (z) D(z)\Big)_{jk}\right| \cdot |X^{\dist_1}_{k\beta}(z)|.
\end{eqnarray}
From Lemma \ref{lem:exp-decay}, there exist positive constants $\gamma\in(0,\gamma^{\rm r})$ and $C$ independent on $z$, $R_1$ and $R_2$ such that
\begin{eqnarray}
\max\Big\{ |X^{R_1}_{jk}(z)|, |X^{R_2}_{jk}(z)|, |E_{jk}^{-1}(z)| \Big\} \le Ce^{-\gamma |j-k|}.
\end{eqnarray}
We then estimate that for any $j,k\in\Lambda_n^{R_1}$
\begin{multline*}
\left|\Big(B(z) E^{-1}(z) D(z)\Big)_{jk}\right| 
\le C \sum_{p,q\in\Lambda_n^{R_2}\backslash\Lambda_n^{R_1}} e^{-\gamma^{\rm r} |j-p|} e^{-\gamma |p-q|} e^{-\gamma^{\rm r} |q-k|} \\[1ex]
\le C \sum_{q\in\Lambda_n^{R_2}\backslash\Lambda_n^{R_1}} e^{-\tilde{\gamma} |j-q|} e^{-\gamma^{\rm r} |q-k|}
\le C e^{-\tilde{\gamma} (R_1-|j-n|)} \sum_{q\in\Lambda_n^{R_2}\backslash\Lambda_n^{R_1}}  e^{-\gamma^{\rm r} |q-k|}
\le C e^{-\tilde{\gamma} (2R_1-|j-n| - |k-n|)},
\end{multline*}
where $\tilde{\gamma}\in(0,\gamma)$ and $C$ do not depend on $z$, $\dist_1$ and $\dist_2$.
Combining this with \eqref{inv-err-eq-1} and \eqref{inv-X-est}, we have
\begin{align}
\label{X-err-eq}
\nonumber
|X_{\alpha \beta}^{\dist_2}(z)-X_{\alpha \beta}^{\dist_1}(z)|  
& \le C \sum_{j,k\in\Lambda_n^{R_1}} e^{-\gamma |j-\alpha|} e^{-\tilde{\gamma} (2R_1-|j-n| - |k-n|)} e^{-\gamma |k-\beta|} 
\\[1ex]
& \le C N_1^2 e^{-\tilde{\gamma}(2 R_1-|\alpha-n|-|\beta-n|)}.
\end{align}
By using the assumption ${\bf (A1)}$,  \eqref{inv-X-est} and \eqref{X-err-eq}, the approximation error for the SPM can thus be estimated by
\begin{align*}
|\tau^{\dist_1}_{nn}(z)-\tau^{\dist_2}_{nn}(z)| 
& =  \left| \sum_{\beta\in\Lambda_n^{R_2}} X_{n\beta}^{\dist_2}(z)  \tau^{{\rm r},R_2}_{\beta n}(z) - \sum_{\beta\in\Lambda_n^{R_1}} X_{n\beta}^{\dist_1}(z)  \tau^{{\rm r}, R_1}_{\beta n}(z) \right| 
\\[1ex]\nonumber
& \le  \sum_{\beta\in\Lambda_n^{R_1}} \left| X_{n\beta}^{\dist_2}(z) - X_{n\beta}^{\dist_1}(z) \right|  \left| \tau^{{\rm r},R_1}_{\beta n}(z) \right| 
+ \sum_{\beta\in\Lambda_n^{R_2} / \Lambda_n^{R_1}} \left| X_{n\beta}^{\dist_2}(z)  \tau^{{\rm r}, R_2}_{\beta n}(z) \right| 
\\[1ex]
& \le C N_1^2 e^{-\tilde{\gamma} R_1} \sum_{\beta\in\Lambda_n^{R_1}} \left| \tau^{{\rm r},R_1}_{\beta n}(z) \right| 
+ C e^{-\gamma R_1} \sum_{\beta\in\Lambda_n^{R_2} / \Lambda_n^{R_1}} \left|  \tau^{{\rm r}, R_2}_{\beta n}(z) \right| 
\\[1ex]
& \le  C e^{-\eta R_1},
\end{align*}
where $C$ and $\eta$ depend on $C^{\rm r}$, $\gamma^{\rm r}$, $\sigma $ and $d$, but not on $z$, $\dist_1$ and $\dist_2$.
%
By taking $R_2\rightarrow\infty$, we can obtain the exponential convergence \eqref{tau-site-truncate-err}.
\end{proof}

Using the SPM approximation \eqref{green-matrix-dyson-eq}, the local DoS in cell $\Omega_n$ can be approximated by the corresponding truncation of $\Lambda$ as
\begin{eqnarray}
\nonumber
\label{approx-ldos_n}
\mathscr{D}^{\dist}_n(f) := \frac{1}{2\pi i} \oint_{\mathscr{C}} f(z) \int_{\Omega_n} G^R(\vr,\vr;z) \dd\vr \dd z,
\end{eqnarray}
where $G^R$ is the corresponding truncated Green's function
\begin{align*}
G^R(\vr,\vr;z) = -i\delta_{nn'} \sqrt{z} \zeta_{n}(\vr,z) \xi_{n}(\vr,z) 
+ \zeta_{n}(\vr;z) \hskip 0.1cm\tau^R_{nn}(z) \hskip 0.1cm \zeta_{n}(\vr;z) .
\end{align*}
Note that the local quantities $\zeta_n(\vr;z)$ and $\xi_{n}(\vr,z)$ are obtained locally in cell $\Omega_n$, and hence will not be affected by the scattering region truncation.
Therefore, the accuracy of the local DoS approximation is completely determined by that of SPM.
We can immediately obtain from Theorem \ref{thm:spm-exp-convergence} that, for $f\in \Sc$ and $n\in\Lambda$,
\begin{eqnarray}
\label{err-site-eq}
\big|\ldos^{\dist}(f)-\ldos(f)\big| 
\le \frac{1}{2\pi} \oint_{\mathscr{C}} |f(z)| \left| \tau^{\dist}_{nn}(z)-\tau_{nn}(z) \right| \int_{\Omega_n} |\zeta_{n}(\vr;z)|^2  \dd\vr \dd z \le C e^{-\eta R}.
\end{eqnarray}

\begin{remark}
[Approximation of the electron density]
\label{remark:converge:density}
A direct application of the approximation \eqref{SPM:R_implicit} is the ground state electron density calculation \eqref{density-eq}.
In cell $\Omega_n~(n\in\Lambda)$, the electron density $\rho(\vr)$ can be approximated with the scattering region cutoff $\Lambda_n^R$
\begin{eqnarray*}
\rho_n^R(\vr) := \frac{1}{2\pi i}\oint_{\mathscr{C}} f_{\rm FD}(z)  G^R(\vr,\vr;z) \dd z \qquad\vr\in\Omega_n.
\end{eqnarray*}
By \eqref{err-site-eq}, we have the following exponential convergence rate $\big\|\rho_n^R-\rho\big\|_{L^1(\Omega_n)} \leq C e^{-\eta R}$.
\end{remark}

\subsection{An iterative scheme}
\label{sec:linear-scaling}

The bottleneck of evaluating the local DoS in cell $\Omega_n$ is to compute the diagonal entry $\tau^R_{nn}(z)$ of the approximate SPM by solving the linear system \eqref{linear-system}.
%
%
In the community of numerical linear algebra, there are many technique that can reduce the computational cost of matrix inversion when the matrix process some kind of sparsity, which have also appeared in electronic structure calculations \cite{benzi13,chen16,lin16,lin2011fast,lin2011selinv,thiess12,zeller08}. 
Following similar ideas, we will discuss how to improve the efficiency of the matrix inversion by exploiting the sparsity of the SPM. 
In the rest of this section, we will develop a fixed-point iteration scheme to solve the linear system, in which the length of scattering path is further truncated by exploiting the matrix sparsity.

Assume that at the $\nu$-th step ($\nu\in\N$), we have $\tau^R_{(\nu)}(z)$ as a guess of the SPM, then the following fixed-point iteration scheme can be derived directly from \eqref{SPM:R_implicit}
\begin{eqnarray}
\label{SPM:iteration}
\tau^{\dist}_{(\nu+1)}(z) = \tau^{{\rm r},R}(z) +  \tau^{{\rm r},R}(z) \hskip 0.1cm \Delta t^{R}(z) \hskip 0.1cm \tau^{\dist}_{(\nu)}(z) .
\end{eqnarray}
Note that we can efficiently obtain the diagonal entry by considering only the $n$-th column of $\tau^R(z)$, as each column in the iterative scheme can be iterated alone.
From the assumption {\bf (A3)} and the fixed-point theorem, we have
\begin{eqnarray}
\label{iter-converge-eq}
\|\tau^{\dist}(z) - \tau_{(\nu)}^{\dist}(z)\|_{1} \le \| \tau^{{\rm r},R}(z) \Delta T^{R}(z) \|_{1}^{\nu} \|\tau^{\dist}(z) - \tau_{(0)}^{\dist}(z)\|_{1} \le \kappa^{\nu} \|\tau^{\dist}(z) - \tau_{(0)}^{\dist}(z)\|_{1}, 
\end{eqnarray}
with $\kappa$ given in \eqref{spectrum-assumption2}.
Consequently, $\tau_{(\nu)}^{\dist}(z)$ can converge to $\tau^{\dist}(z)$ as $\nu\rightarrow \infty$.
In practical calculations, we will take the maximum iteration time $N_{\rm i}$ sufficiently large to ensure the convergence, and choose the initial guess of SPM randomly or by taking $\tau_{(0)}^{R}(z)=\tau^{{\rm r},\dist}(z)$ (that corresponds to the first term of Born series, see the following remark).

\begin{remark}
[Born series expansion]
\label{remark:born}
Under the assumption {\bf (A3)}, the approximate SPM $\tau^R(z)$ can be described by the Born series \cite{gonis00,mavropoulos06} via the expansion of \eqref{SPM:R_implicit}
\begin{align}
\label{born-eq}
\nonumber
\tau^{\dist}_{jk}(z) = \tau^{{\rm r}, R}_{jk}(z)
& + \sum_{p\in\Lambda^{\dist}_n}  \tau^{{\rm r},R}_{jp}(z) \Delta t_{p}(z) \tau^{{\rm r},R}_{pk}(z) 
\\[1ex]
& + \sum_{p,q\in\Lambda^{\dist}_n} \tau^{{\rm r},R}_{jp}(z) \Delta t_{p}(z) \tau^{{\rm r},R}_{pq}(z) \Delta t_{q}(z) \tau^{{\rm r},R}_{qk}(z) + \cdots 
\qquad\forall~j,k\in\Lambda_n^R.
\end{align}
From a physical point of view, \eqref{born-eq} characterizes the multiple scattering behaviors of electron from site $j$ to $k$, where $\Delta t_{k}(z)$ and $\tau^{{\rm r},\dist}_{jk}(z)$ can be interpreted as local scattering event at site $k$ and electron traveling under the reference background from site $j$ to $k$, respectively.
Note that assumption {\bf (A3)} indicates the amplitude of each scattering event, combined with an electron traveling under the reference background, to be small.
Hence, it is natural to truncate the ``times" of scattering event by a finite number,
which corresponds to the truncation of the Born series \eqref{born-eq}.
In fact, it can be seen that the $\Nit$-term truncation of Born series is essentially equivalent to the fixed-point iteration scheme \eqref{SPM:iteration}, by choosing the initial guess $\tau^R_{(0)}(z)=\tau^{{\rm r}, R}(z)$. 
\end{remark}

We then exploit the exponential off-diagonal decay of $\tau^{{\rm r},R}(z)$ to make further approximation of the iterative scheme \eqref{SPM:iteration}.  
Let $\Ntr\in(0,R)$ denote the parameter for truncating the length of reference SPM path.
We approximate $\tau^{{\rm r},R}(z)$ by
\begin{eqnarray}
\label{numer-truncation-sparse-def}
\tau^{{\rm r},R\Ntr}_{jk}(z) := \left\{
\begin{array}{ll}
\tau^{{\rm r},R}_{jk}(z) &\qquad |j-k|\le \Ntr, \\[1ex]
0 &\qquad {\rm otherwise}.
\end{array}\right.
\end{eqnarray} 
Then the iterative scheme \eqref{SPM:iteration} can be approximated by
\begin{eqnarray}
\label{trun-iter-eq}
\tau^{\dist\Ntr}_{(\nu+1)}(z) = \tau^{{\rm r},R\Ntr}(z) + \tau^{{\rm r},R\Ntr}(z) \hskip 0.1cm \Delta t^{R}(z) \hskip 0.1cm \tau^{R \Ntr}_{(\nu)}(z) \qquad 0\leq \nu\leq \Nit-1,
\end{eqnarray}
where the initial vector is chosen by  $\tau_{(0)}^{R\Ntr}(z) = \tau_{(0)}^{R}(z)$.
Note that the fixed-point iteration \eqref{trun-iter-eq} can also converge since the following estimate is satisfied
$$
\| \tau^{{\rm r}, R\Ntr}(z) \Delta t^{\dist}(z) \|_{1} \le  \| \tau^{{\rm r},R}(z) \Delta t^{\dist}(z) \|_{1} \le \kappa .
$$

We provide some intuitive examples to interpret the iterative scheme \eqref{trun-iter-eq} and to explain how the truncation by $L$ can correspond to a cutoff of the scattering path.
We show five different scattering paths in Figure \ref{fig-scattering-path}, which correspond to the terms in Born series in Remark \ref{remark:born}.
There are three features required in the iterative scheme \eqref{trun-iter-eq}: (a) all the scattering paths are confined in the bounded scattering region $\Lambda^{\dist}_n$ centered at site $n$; (b) each path traverses fewer than $\Nit$ sites, starting from site $n$ and finally returning to $n$; and (c) the length of each segment of the paths are less than the truncated length of scattering path $\Ntr$.
For the given numerical parameters, the first three paths in the figure are included in \eqref{trun-iter-eq} (or \eqref{born-eq}), while the other two are excluded.

\begin{figure}[htbp]
	\centering
	\includegraphics[width=10.cm]{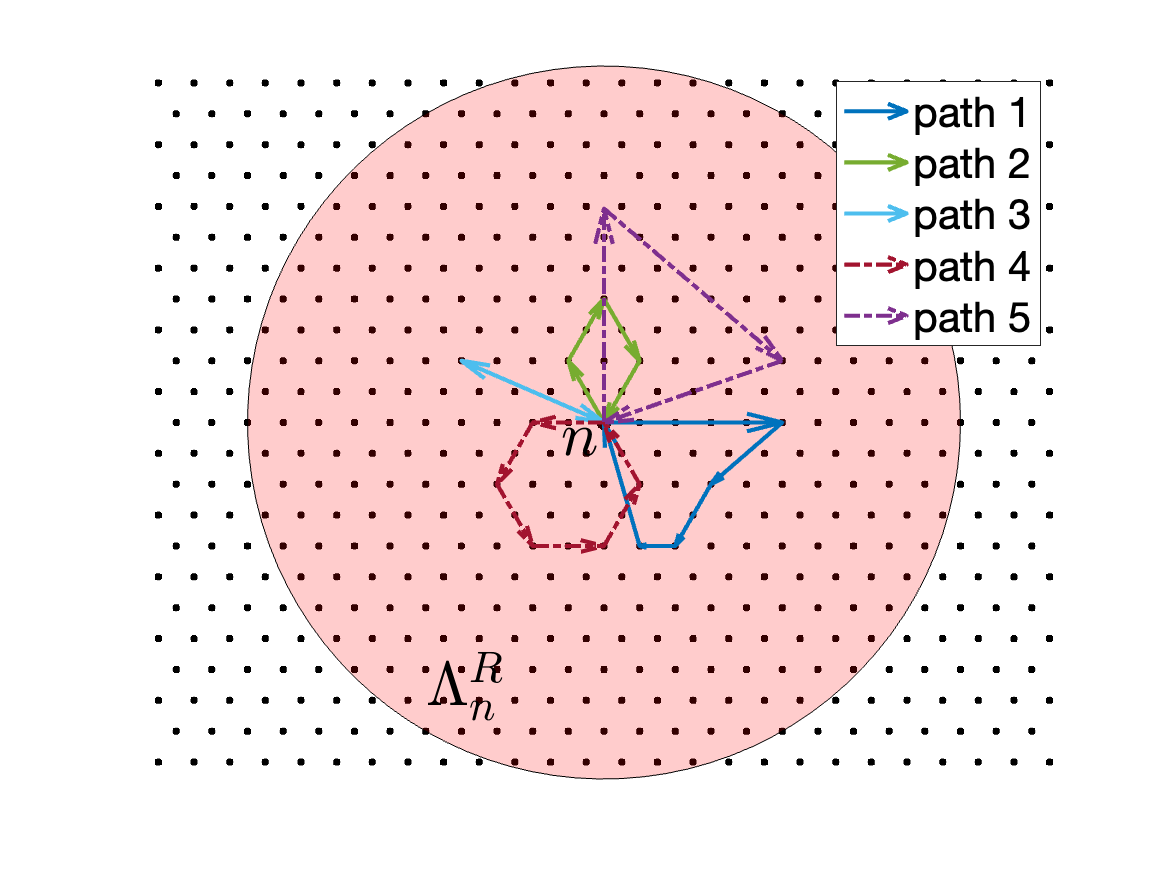}
	\vskip -0.5cm
	\caption{A schematic illustration of the  scattering paths when $R=10, \Nit=5, \Ntr=5$. Path 1-3 satisfy (a)-(c), while path 4 and path 5 do not satisfy (b) and (c), respectively.}
	\label{fig-scattering-path}
\end{figure}

Compared to a direct inversion with $\mathcal{O}(R^{3d})$ operations, the iterative scheme \eqref{trun-iter-eq} reduces the cost to $\mathcal{O}(N_{\rm i} L^d \dist^d)$. 
Moreover, we can choose a sparse initial guess $\tau_{(0)}^{\dist}(z)$ to ensure that $\tau_{(\nu)}^{\dist\Ntr}(z)$ remains sparse as well.
This results in a computational effort of $\mathcal{O}(N^2_{\rm i} L^{2d})$, which is independent on the size of the scattering region $R$.
The following theorem provides approximate errors of the SPM with respect to the numerical parameters $\dist$, $\Ntr$ and $\Nit$.

\begin{theorem}
\label{thm:iter-convergence}
Assume that the reference system satisfies the conditions {\bf (A1)} and {\bf (A3)}. 
Then for any $n\in\Lambda$ and $z\in\C$ satisfying \eqref{resolvant-dist}, there exist positive constants $C$ and $\eta$ that are independent on $z$, $\dist$, $\Nit$ and $\Ntr$, such that
\begin{align}
\label{convergence:spm:path}
\left| \Big(\tau^{\dist\Ntr}_{(\Nit)}(z)\Big)_{nn}-\tau_{nn}(z) \right| 
\leq C \left( e^{-\eta\dist} + \kappa^{N_{\rm i}} + \frac{1-\kappa^{\Nit}}{1-\kappa}e^{-\gamma^{\rm r} \Ntr} \right) .
\end{align}
\end{theorem}

\begin{proof}	
We first show the boundness of $\tau_{(\nu)}^{\dist}(z)$.
Due to the convergence \eqref{iter-converge-eq}, we have $\|\tau_{(\nu)}^{R}(z)\|_{1} \le C$ for $1\le \nu\le \Nit$ and $z\in\mathscr{C}$, where $C$ depends on $\|\tau^{\dist}(z)\|_{1}$ and $\|\tau_{(0)}^{\dist}(z)\|_{1}$.
We mention that $\|\tau^{\dist}(z)\|_{1}$ is bounded by some positive constant independent on the truncation $\dist$ and $z$ due to the exponentially off-diagonal decay of $\tau^{\dist}(z)$ in \eqref{inv-X-est}.

By subtracting \eqref{SPM:iteration} from \eqref{trun-iter-eq}, we have
\begin{align*}		
\tau_{(\nu)}^{\dist\Ntr}(z) - \tau_{(\nu)}^{\dist}(z)
& = \Big(\tau^{{\rm r},\dist\Ntr}(z) - \tau^{{\rm r},\dist}(z) \Big) \Big( I +\Delta t^{\dist}(z) \tau_{(\nu-1)}^{\dist}(z) \Big) 
\\[1ex]
& \qquad \qquad
+ \tau^{{\rm r},\dist\Ntr}(z)  \Delta t^{\dist}(z) \Big( \tau_{(\nu-1)}^{\dist\Ntr}(z) - \tau_{(\nu-1)}^{\dist}(z) \Big).
\end{align*}
From {\bf (A3)}, \eqref{numer-truncation-sparse-def} and the boundness of $\tau_{(\nu)}^{\dist}(z)$, it then follows that
\begin{align*}
\left\| \tau_{(\nu)}^{\dist\Ntr}(z) - \tau_{(\nu)}^{\dist}(z) \right\|_1
& \leq Ce^{-\gamma^{\rm r} \Ntr} + \kappa 
\| \tau_{(\nu-1)}^{\dist\Ntr}(z) - \tau_{(\nu-1)}^{\dist}(z) \|_1 
\\[1ex] 
& \leq C e^{-\gamma^{\rm r} \Ntr} (1+\kappa+\cdots+\kappa^{\nu-1})
\leq C \frac{1-\kappa^{\nu}}{1-\kappa} e^{-\gamma^{\rm r} \Ntr} 
\qquad 1\le \nu \le\Nit,   
\end{align*}
where $C$ depends on the dimension $d$, the local potential $v_n$, the reference potential $v_n^{\rm r}$ and the initial guess.
Combining this with \eqref{tau-site-truncate-err} and \eqref{iter-converge-eq}, we obtain
\begin{eqnarray}\nonumber
\left| \Big(\tau^{\dist\Ntr}_{(\Nit)}(z)\Big)_{nn}-\tau_{nn}(z) \right| 
\le C \left( e^{-\eta\dist} +  \kappa^{\Nit} + \frac{1-\kappa^{\Nit}}{1-\kappa} e^{-\gamma^{\rm r} \Ntr} \right),
\end{eqnarray}
where $C$ depends on $C^{\rm r}$, $\gamma^{\rm r}$, $\sigma$, $d$, $v_n$, $v_n^{\rm r}$ and the initial guess, but not on $z$, $\dist$, $\Nit$ and $\Ntr$.
This completes the proof.
\end{proof}

We then approximate the local DoS on cell $\Omega_n$ by
\begin{eqnarray*}
\mathscr{D}^{\dist\Nit\Ntr}_n(f) := \frac{1}{2\pi i} \oint_{\mathscr{C}} f(z) \int_{\Omega_n} -i\delta_{nn'} \sqrt{z} \zeta_{n}(\vr,z) \xi_{n}(\vr,z) + \zeta_{n}(\vr;z) \Big(\tau^{\dist\Ntr}_{(\Nit)}(z)\Big)_{nn} \zeta_{n}(\vr;z) \dd\vr \dd z.
\end{eqnarray*}
An immediate consequence of Theorem \ref{thm:iter-convergence} is the following convergence of local DoS with respect to the numerical parameter $\dist$, $\Nit$ and $\Ntr$
\begin{eqnarray*}
|\ldos^{\dist \Nit \Ntr}(f) - \ldos(f)| \leq C \left( e^{-\eta\dist} + \kappa^{\Nit} + \frac{1-\kappa^{\Nit}}{1-\kappa} e^{-\gamma^{\rm r} \Ntr} \right) .
\end{eqnarray*}

\section{Numerical experiments}
\label{sec:numerics}
\setcounter{equation}{0} \setcounter{figure}{0}

In this section we will present the numerical experiments of some typical systems to support the theory.
All simulations are performed on a PC with Intel Core i9-2.3 GHz and 16GB RAM, by using the Matlab package {\em ScreenedKKR} \cite{github}.
To test the numerical errors, the results obtained by using sufficiently large discretizations are taken to be the exact ones.
We use the Fermi-Dirac function to measure the DoS, that is, we will test the convergence of the electron density approximations. 
To efficiently evaluate the local DoS, an alternative contour in \ref{append:contour} is used.
We point out that the fixed-point iteration scheme \eqref{SPM:iteration} would break down if assumption {\bf (A3)} is not satisfied for a few energy parameters on $\mathscr{C}$.
In our numerical experiments, the TFQMR solver is applied for such cases.

We will use the rectangular potentials \eqref{ref-potential-def} as screened potentials to construct the reference systems.
In particular, we compare different screened potentials $V^{\rm r}_1$ and $V^{\rm r}_2$ with the height and range parameters $(W,H)=(0.12, 10)$ and $(W,H)=(0.15, 20)$, respectively.

\vskip 0.3cm

\noindent
{\bf Example 1.} (periodic system)
We consider a periodic system in $\R$ with lattice $\Lambda=\Z$, 
where the potential is given by
\begin{eqnarray}\nonumber
V(x) = a \sum_{k\in\Lambda} e^{-\frac{|x-k|^2}{b^2}} \qquad x\in\R, \qquad \text{with}~a=10, ~b=0.1.
\end{eqnarray}
We first display the off-diagonal decay of reference SPM at the first Matsubara energy point $z_1$ (defined in \eqref{matsubara-def}) in Figure \ref{fig-ref_decay}.
As mentioned in Remark \ref{remark:tau-decay}, the SPM has the slowest off-diagonal decay rate at this point.
We observe that using the screened potential can give a much faster off-diagonal decay rate of the SPM, compared with the case of free-electron gas.

We then present the convergence of $\tau_{nn}^R(z)$ at the first and second Matsubara energies $z_1$ and $z_2$ in Figure \ref{fig-X_1d}, and the convergence of local DoS in Figure \ref{fig-ldos_1d}.
We observe that the errors for both quantities decay exponentially fast as the size of scattering region $R$ increases, which is consistent with our theory.
The pictures also suggest that the convergence is not sensitive to the choice of screened potentials.

\begin{figure}[htb!]
\centering
\subfloat[]{
\includegraphics[width=0.33\textwidth]{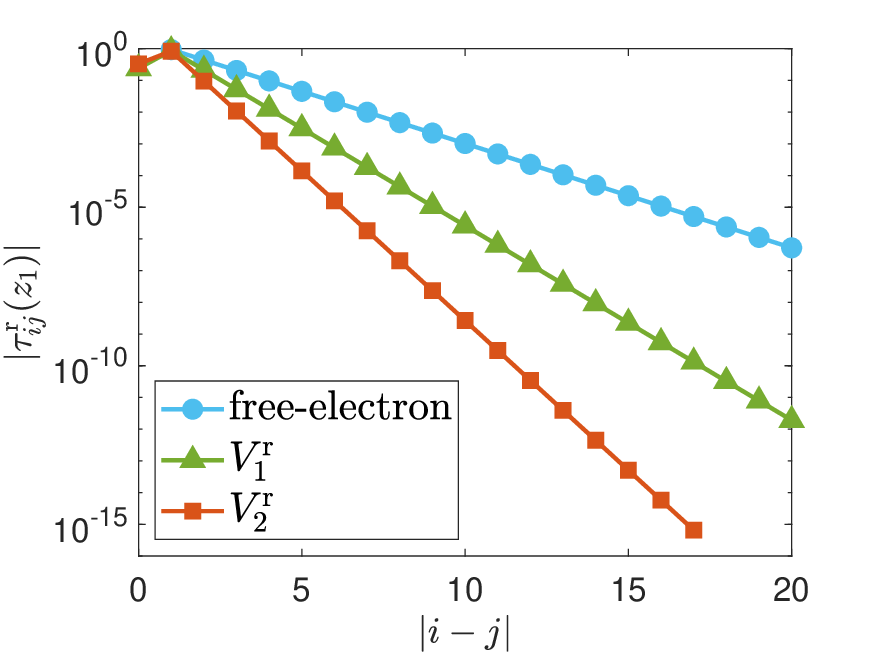}
\label{fig-ref_decay}
}
\subfloat[]{
\includegraphics[width=0.33\textwidth]{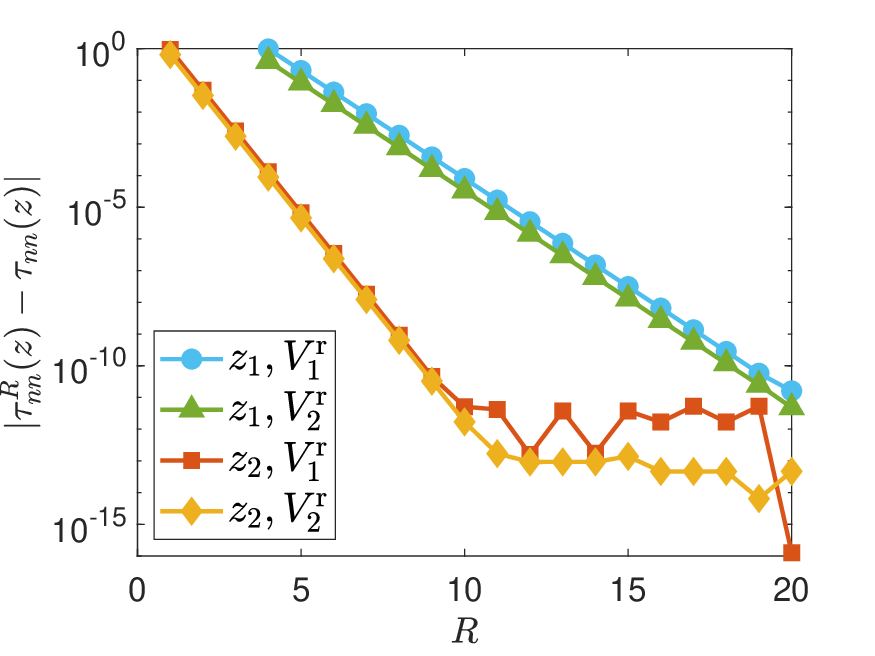}
\label{fig-X_1d}
}
\subfloat[]{
\includegraphics[width=0.33\textwidth]{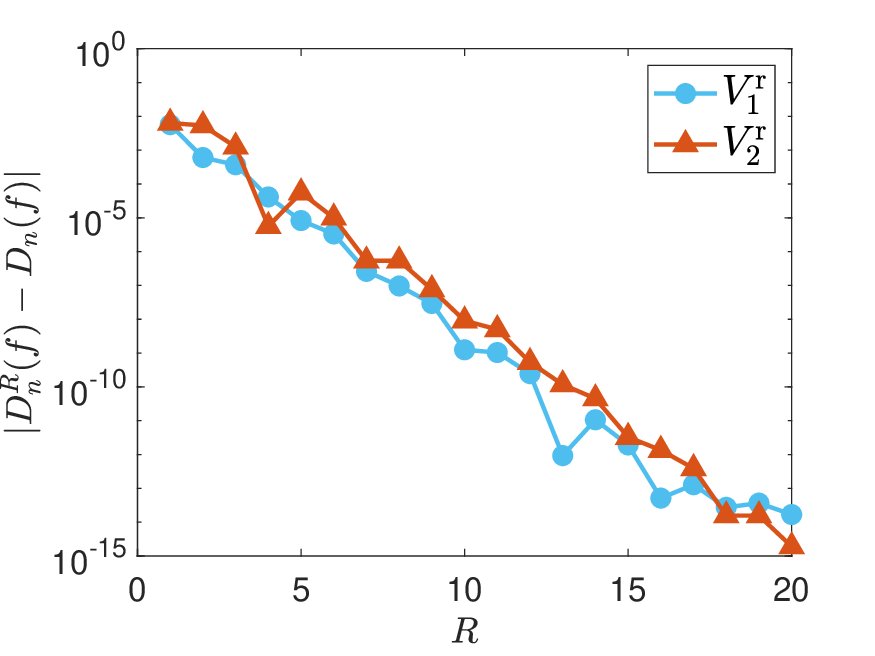}
\label{fig-ldos_1d}
}	
\caption{Example 1: (a) Off-diagonal decay of reference SPM; (b) Convergence of SPM;  (c) Convergence of local DoS.}
\end{figure}

We then test the convergence of various MST algorithms.
For the fixed-point iteration, we only take $V^{\rm r}_1$ as the screened potential as {\bf (A3)} is a necessary condition for this scheme.
We present the convergence of $\tau_{nn}^R(z)$ at the first and second Matsubara energies and the local DoS with respect to the number of iterations $\Nit$ in Figure \ref{fig-tau-it} and Figure \ref{fig-converg_PNit_1d} when $R$ and $\Ntr$ are sufficiently large.
We see similar convergence rates with different choices of initial guess, and find that the initial guess $\tau^{{\rm r}, R}(z)$ motivated from the conception of Born series is slightly better than random chosen one.
We show in Figure \ref{fig-tau-tr} and Figure \ref{fig-converg_Ntr_1d} the convergence of $\tau_{nn}^R(z)$ and local DoS with respect to scattering length of reference $\Ntr$ when $R$ and $\Nit$ are sufficiently large.
We see that both quantities give exponential convergence rates, which are consistent with our theory.
Besides the simple fixed-point iteration, we also implement the TFQMR iterative scheme in Figure $4.4$ and Figure $4.5$.
Compared to Figure \ref{fig-tau-it} and Figure \ref{fig-converg_PNit_1d}, it is shown in Figure \ref{fig-tau-tfqmr-it} and Figure \ref{fig-converg_Nit_tfqmr_1d} that the TFQMR solver can also accelerate the convergence for $V^{\rm r}_1$ with respect to the number of iterations.

\begin{figure}[htb!]
\centering
\subfloat[]{	\includegraphics[width=0.4\textwidth]{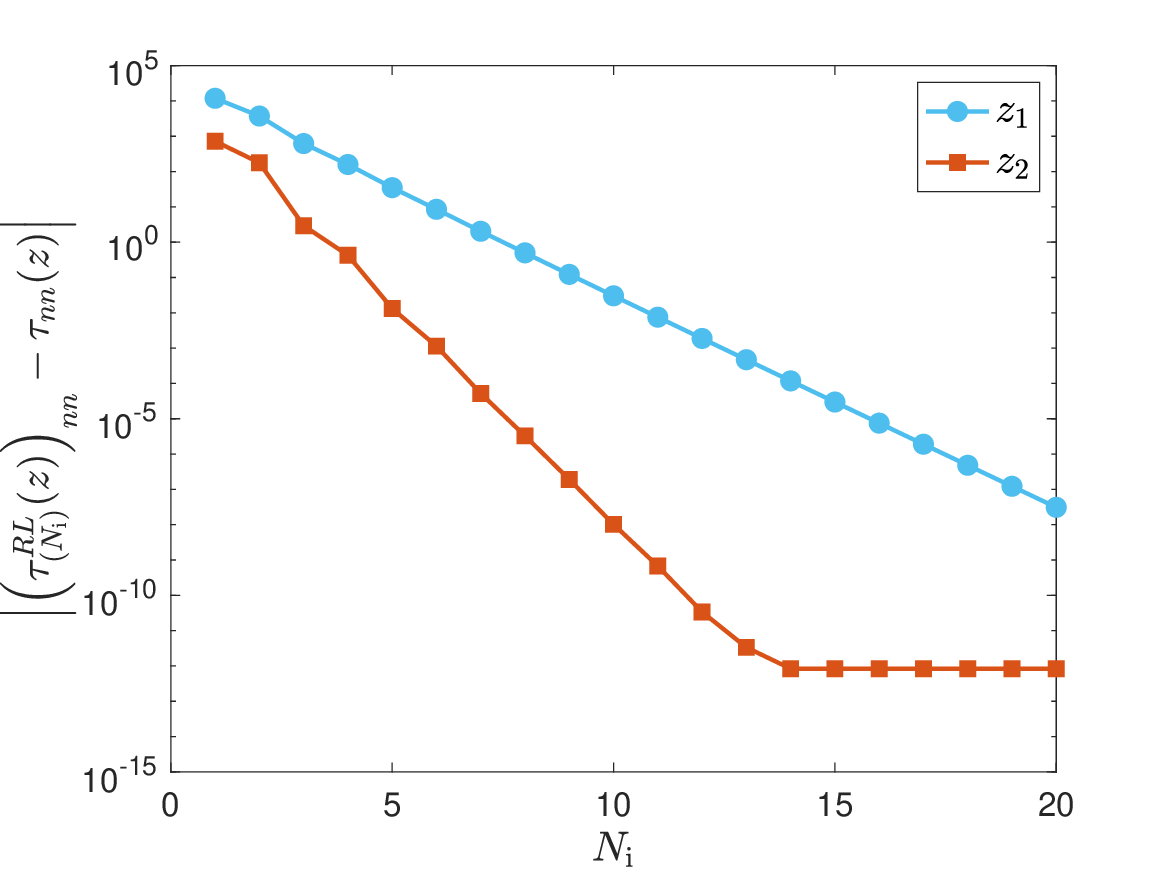}
	\label{fig-tau-it}
}
\subfloat[]{	\includegraphics[width=0.4\textwidth]{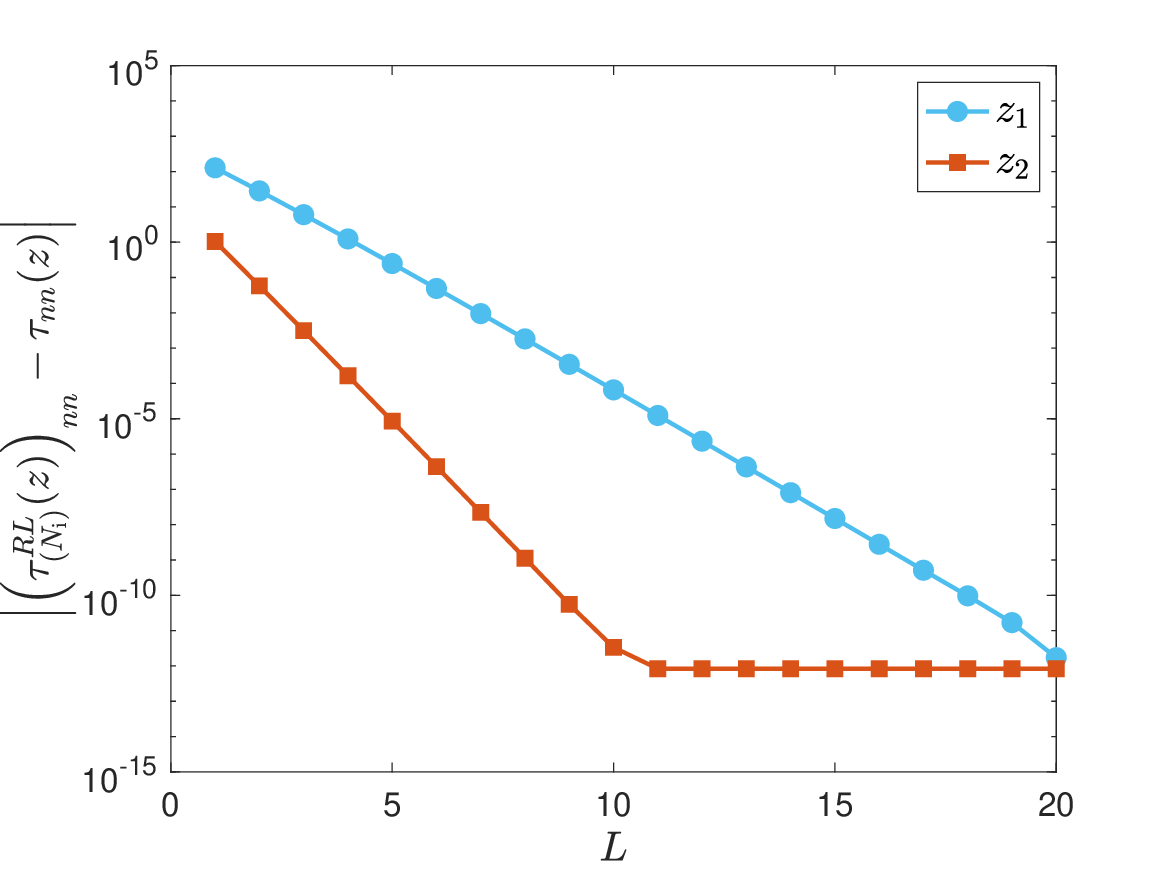}
	\label{fig-tau-tr}
}	
\caption{Example 1: Convergence of SPM by fixed-point iteration. (a) Error w.r.t. number of iterations $\Nit$; (b) Error w.r.t. scattering length of reference $\Ntr$.}
\end{figure}

\begin{figure}[htb!]
\centering
\subfloat[]{
	\includegraphics[width=0.4\textwidth]{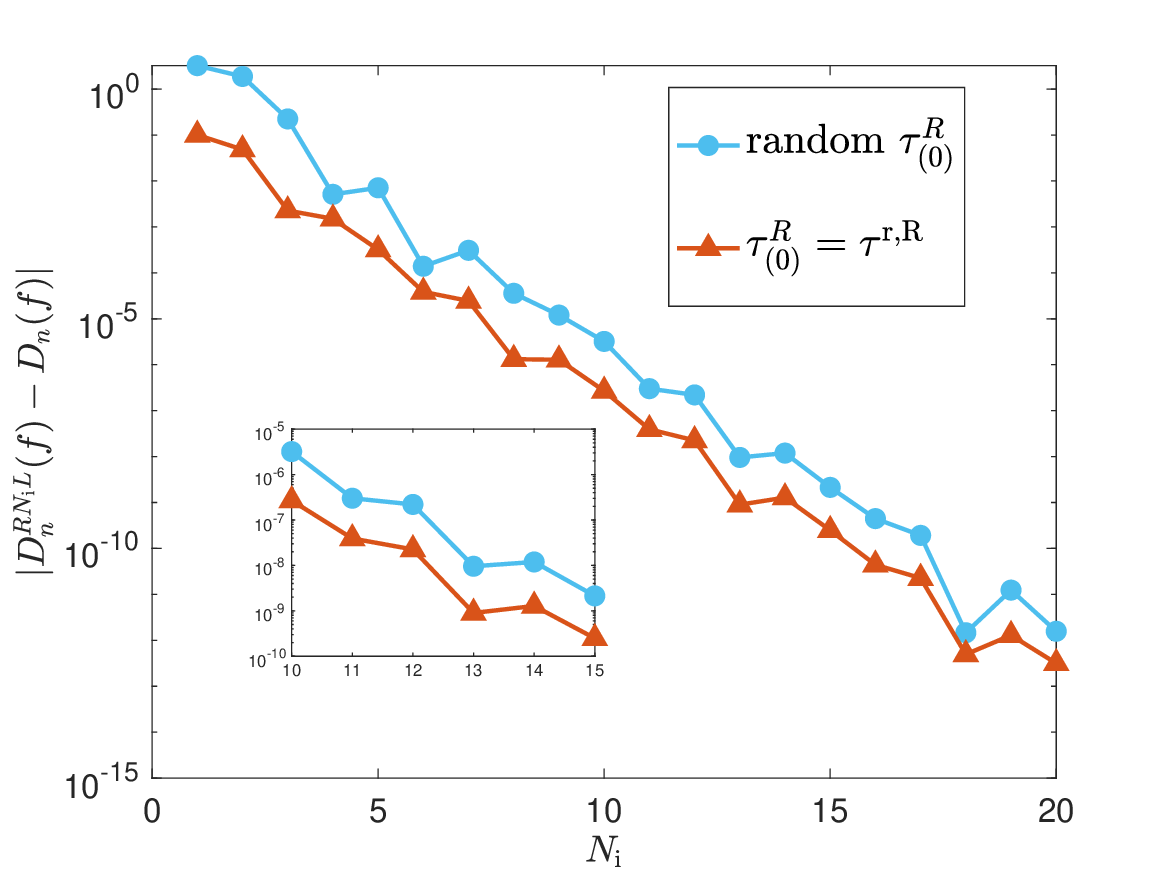}
	\label{fig-converg_PNit_1d}
}
\subfloat[]{
	\includegraphics[width=0.4\textwidth]{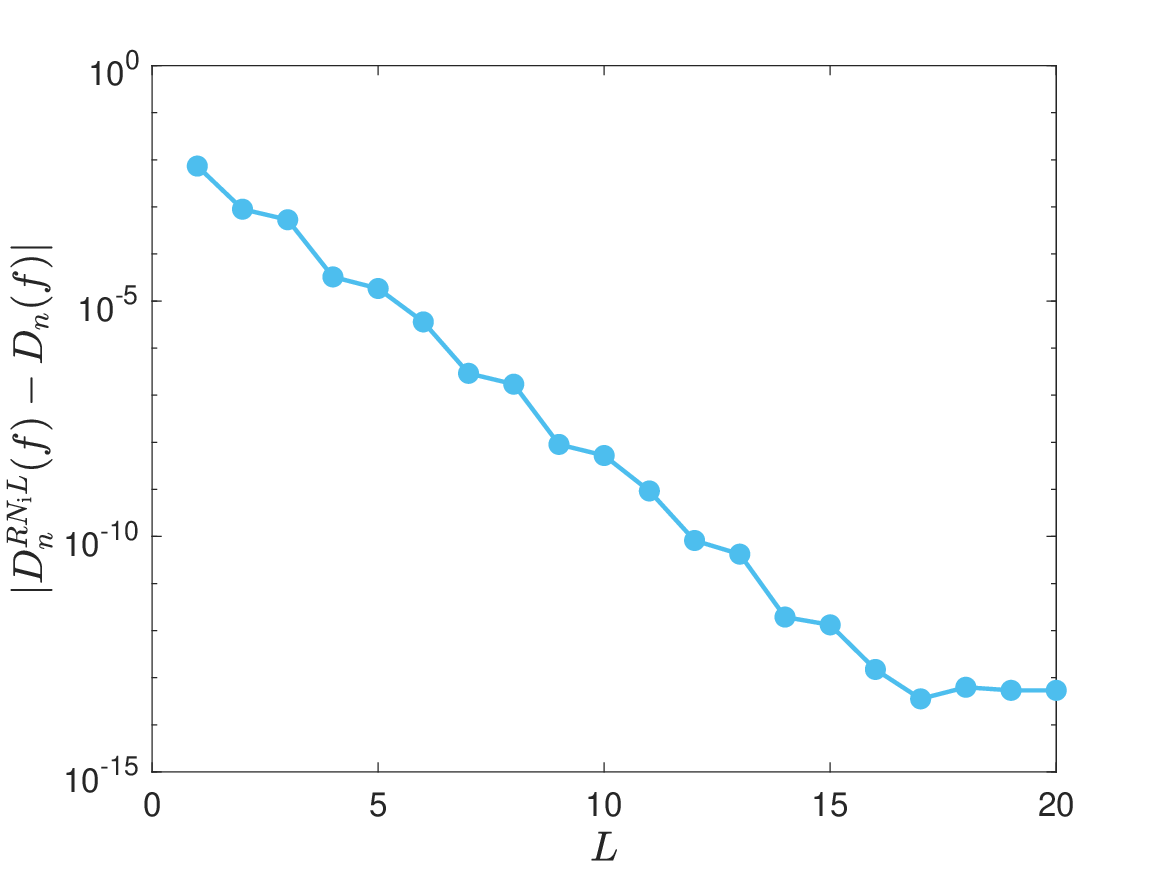}
	\label{fig-converg_Ntr_1d}
}	
\caption{Example 1: Convergence of local DoS by fixed-point iteration. (a) Error w.r.t. number of iterations $\Nit$; (b) Error w.r.t. scattering length of reference $\Ntr$.}
\end{figure}
\begin{figure}[htb!]
\centering
\subfloat[]{	\includegraphics[width=0.4\textwidth]{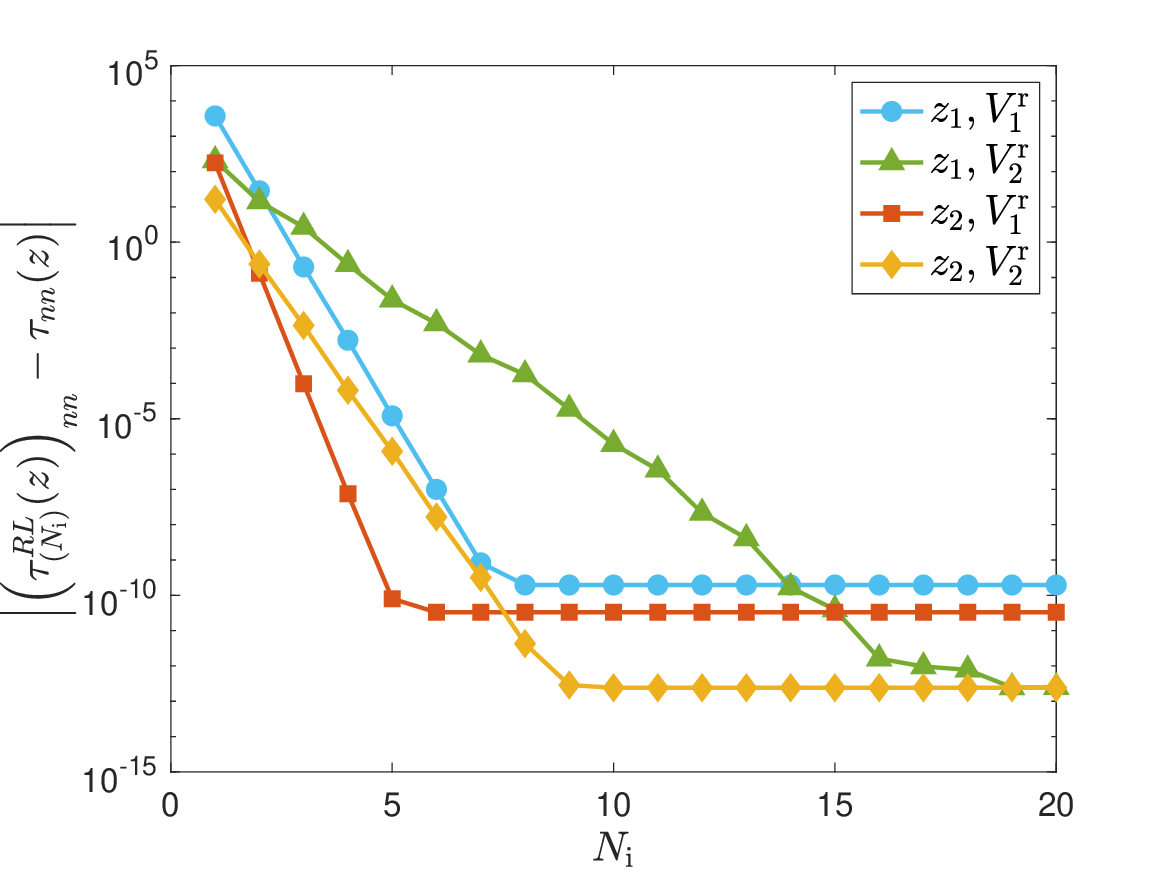}
	\label{fig-tau-tfqmr-it}
}
\subfloat[]{	\includegraphics[width=0.4\textwidth]{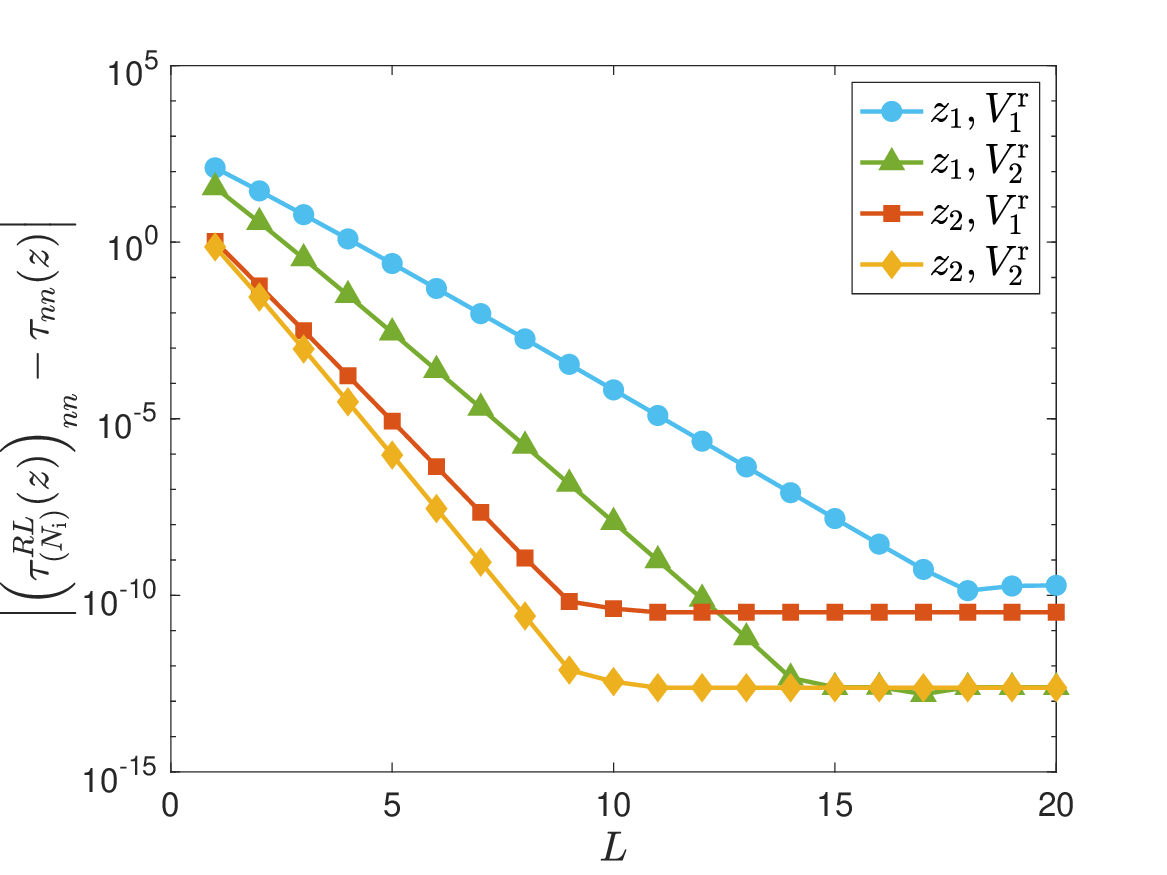}
	\label{fig-tau-tfqmr-tr}
}	
\caption{Example 1: Convergence of SPM by TFQMR iteration. (a) Error w.r.t. number of iterations $\Nit$; (b) Error w.r.t. scattering length of reference $\Ntr$.}
\end{figure}
\begin{figure}[htb!]
\centering
\subfloat[]{
\includegraphics[width=0.4\textwidth]{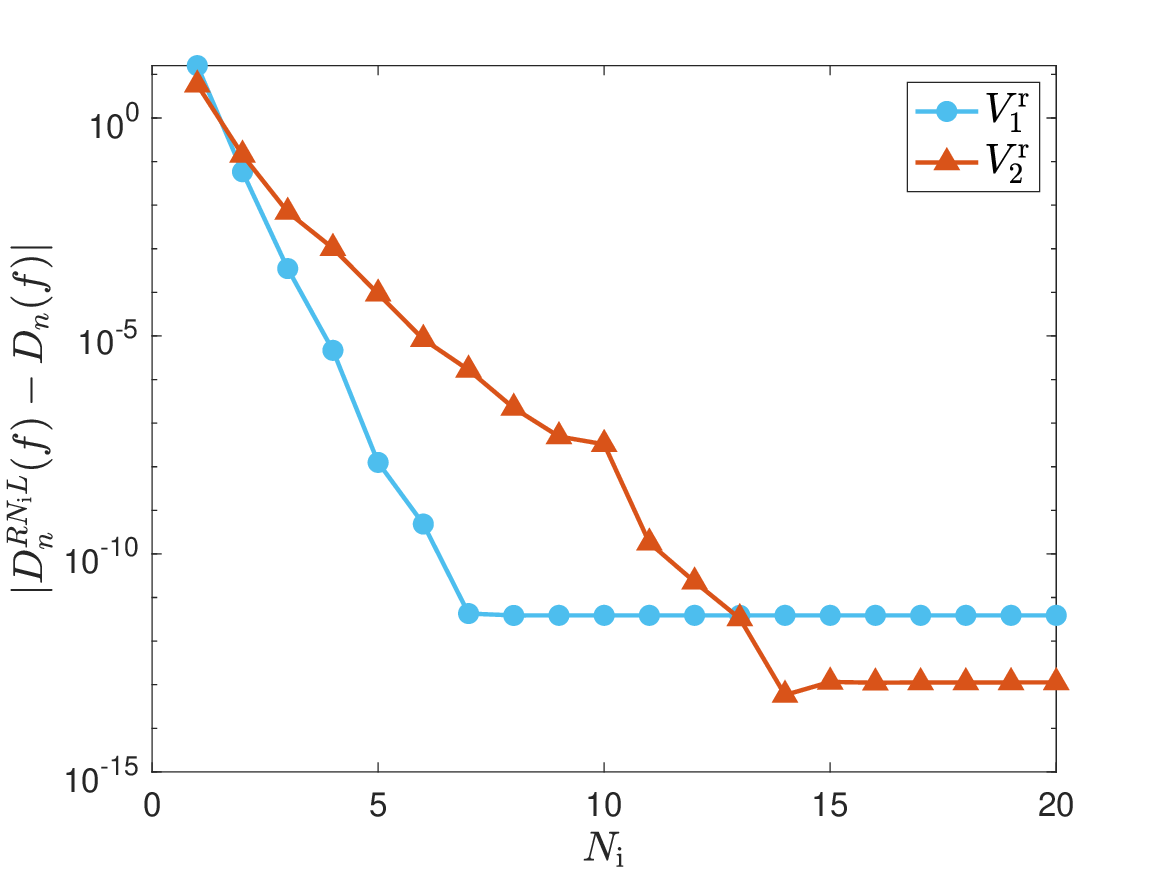}
\label{fig-converg_Nit_tfqmr_1d}
}
\subfloat[]{
\includegraphics[width=0.4\textwidth]{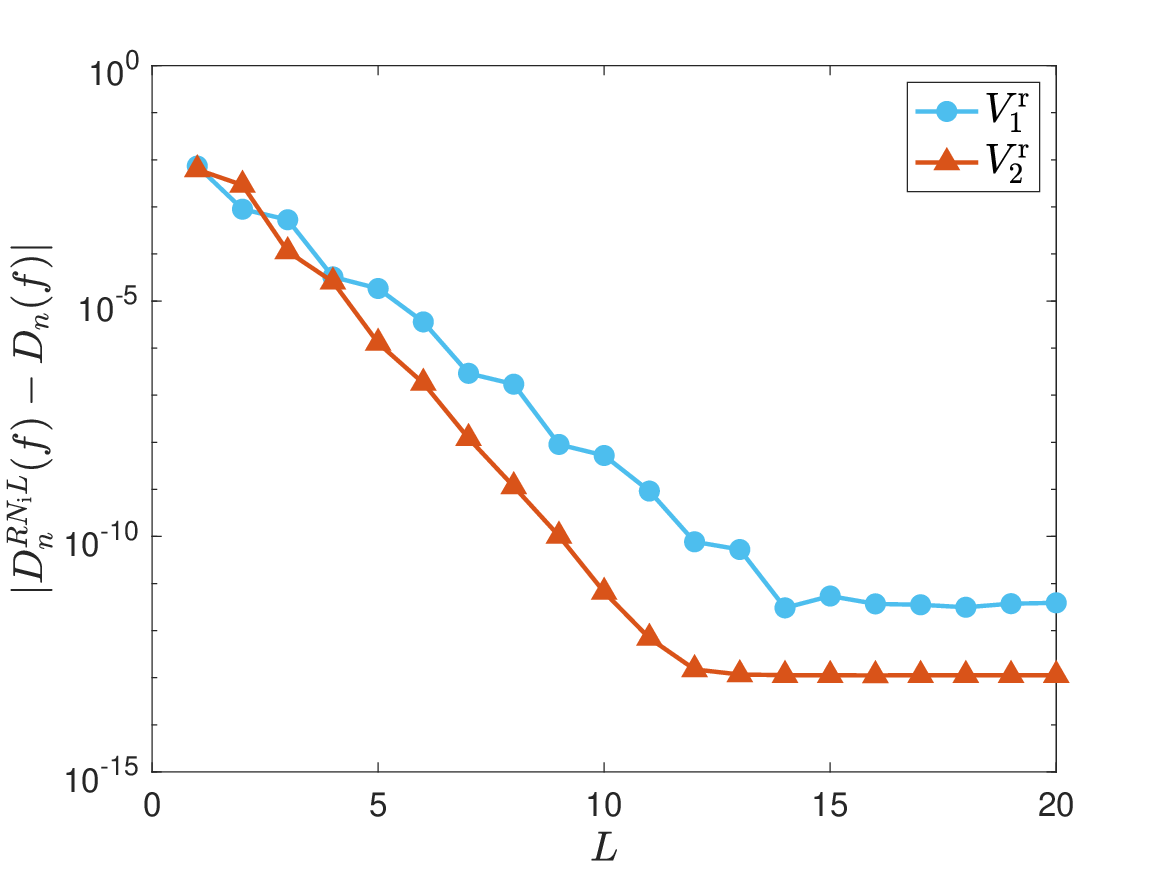}
\label{fig-converg_Ntr_tfqmr_1d}
}	
\caption{Example 1: Convergence of local DoS by TFQMR iteration.            (a) Error w.r.t. number of iterations $\Nit$; (b) Error w.r.t. scattering length of reference $\Ntr$.}
\end{figure}

\vskip 0.3cm

\noindent
{\bf Example 2.} (two-component alloy)
We consider a $A_{1/2}B_{1/2}$-type alloy in $\R$ with lattice $\Lambda=\Z$, where two species of atom $A$ and $B$ randomly occupies each site with equal possibility, respectively.
The potentials of $A$ and $B$ in $\Omega_k~(k\in\Lambda)$ are given by
\begin{eqnarray}\nonumber
V_A(x) = \frac{a}{\sqrt{|x-k|^2+1}},\quad V_B(x) = \frac{b}{\sqrt{|x-k|^2+1}}
\qquad \text{with}~ a=-1, ~b=-2.
\end{eqnarray}
We display the exponential convergence of  local DoS in Figure \ref{fig-ldos_coul_1d}, which shows barely relevance to the choice of reference.
Furthermore, the exponential convergence of TFQMR iteration are observed in Figure \ref{fig-Nit_coul_1d} and \ref{fig-Ntr_coul_1d}, which are consistent with our theoretical results. 

\begin{figure}[htb!]
\centering
\subfloat[]{
	\includegraphics[width=0.33\textwidth]{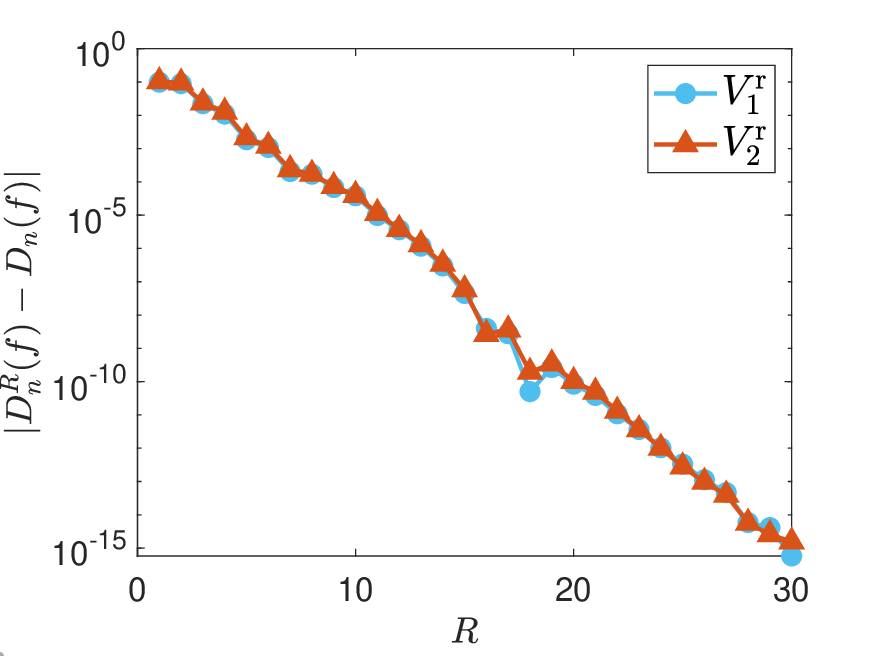}
	\label{fig-ldos_coul_1d}
}	
\subfloat[]{
	\includegraphics[width=0.33\textwidth]{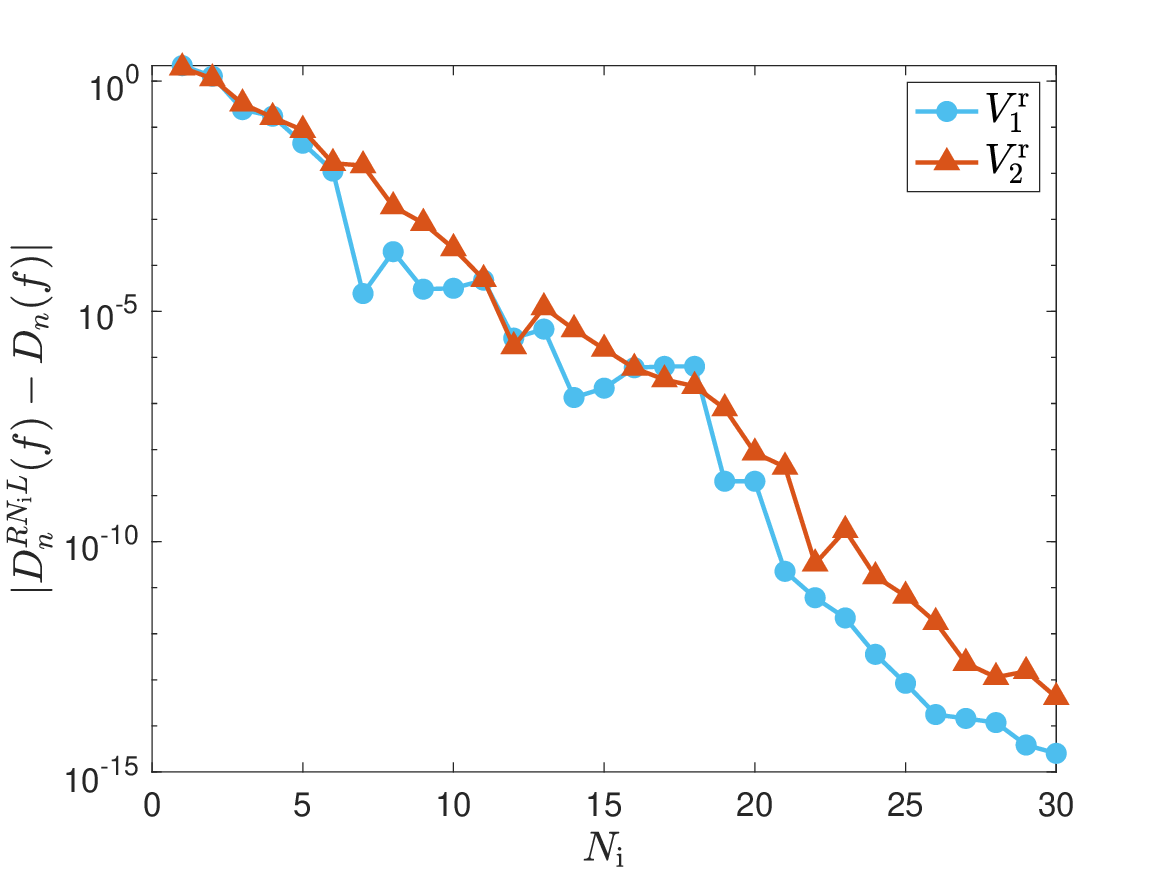}
	\label{fig-Nit_coul_1d}
}
\subfloat[]{
	\includegraphics[width=0.33\textwidth]{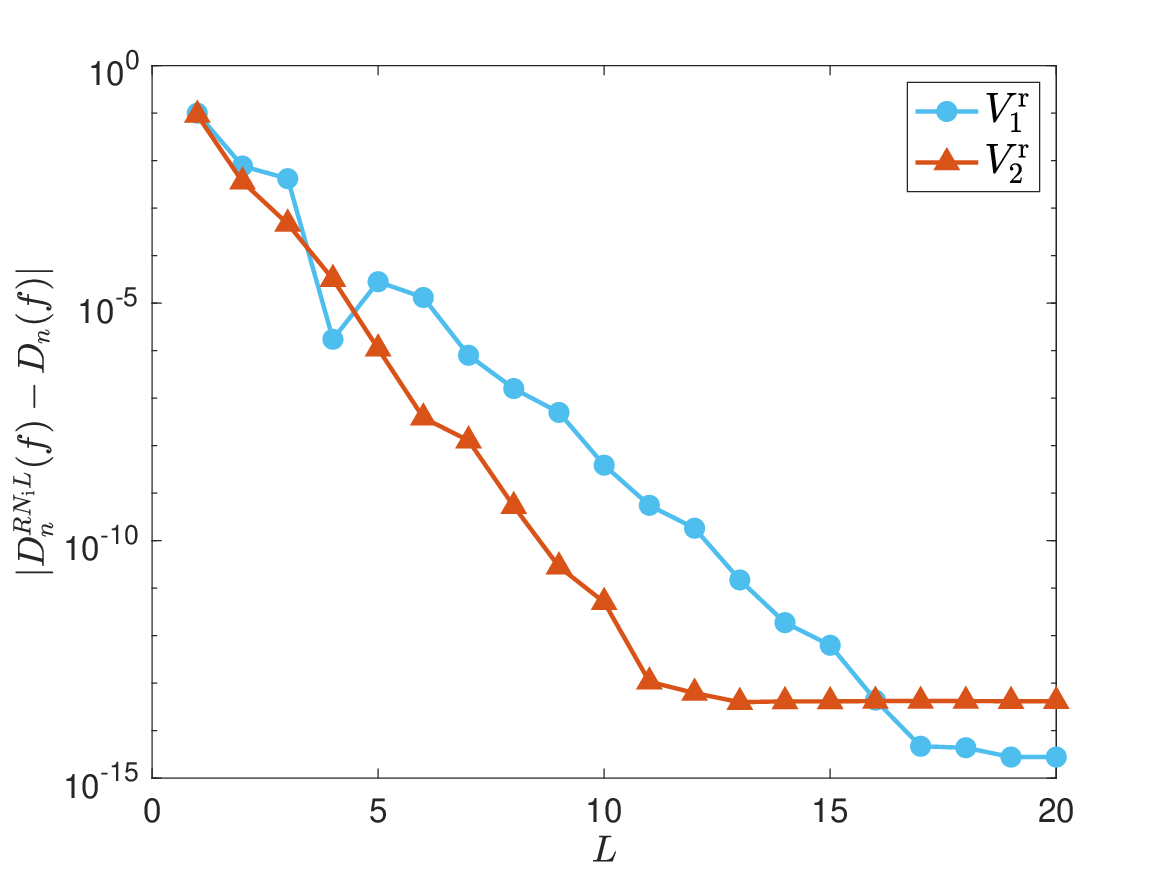}
	\label{fig-Ntr_coul_1d}
}	
\caption{Example 2: Convergence of local DoS. (a) Error w.r.t. size of scattering region $R$; (b) Error w.r.t. number of iterations $\Nit$; (c) Error w.r.t. scattering length of reference $\Ntr$.}
\end{figure}

\vskip 0.3cm

\noindent
{\bf Example 3.} (random potential)
We consider a quasi-1D system in $\R^3$ with the lattice $\Lambda=\{(n,0,0)\}_{n\in\Z}$ with a random potential.
The potential is given by
\begin{eqnarray}\nonumber
V(\vr) = \left\{
\begin{array}{ll}
\displaystyle a_k e^{-\frac{|x-k|^2}{b^2}} \qquad &\vr\in\bar{\Omega}_k \quad \text{for}~k\in\Lambda, \qquad \text{with} ~b=0.1, \\[2ex]
0 & \text{otherwise},
\end{array}\right.
\end{eqnarray}
where $\bar{\Omega}_k=\{\vr\in\R^3: |\vr-k|\le R_{\rm mt}\}$ represents a muffin-tin ball with $R_{\rm mt}=0.3$,
and the random parameters $a_k~(k\in\Lambda)$ are independently and uniformly distributed in $(0,20)$.
We present in Figure \ref{fig-ldos_3d} the exponential convergence of local DoS in a single cell, and further compare the convergence of TFQMR iteration by taking different angular momentum cutoffs in Figure \ref{fig-Nit_3d} and Figure \ref{fig-Ntr_3d}.
It can be seen that the truncation of the angular momentum has only a tiny effect on the convergence.

\begin{figure}[htbp!]
\centering
\subfloat[]{
	\includegraphics[width=0.33\textwidth]{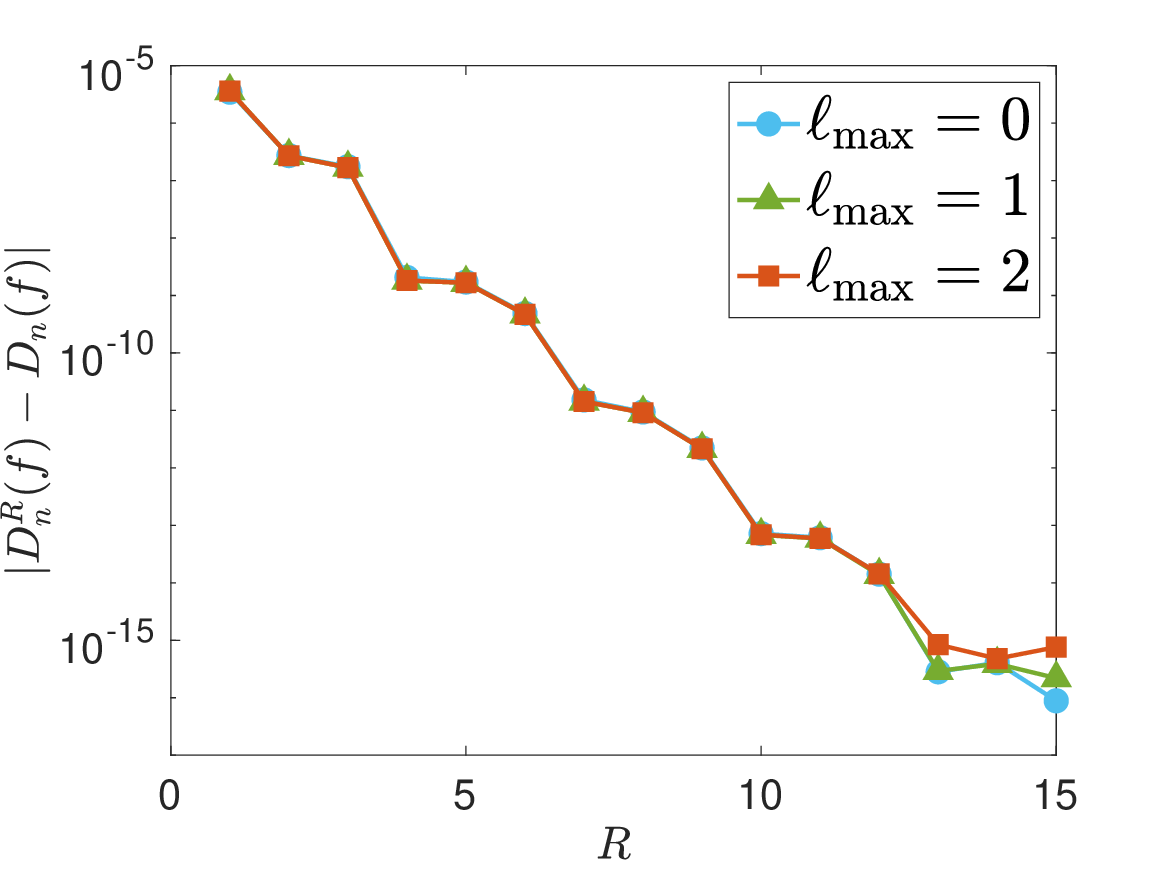}
	\label{fig-ldos_3d}	
}
\subfloat[]{
	\includegraphics[width=0.33\textwidth]{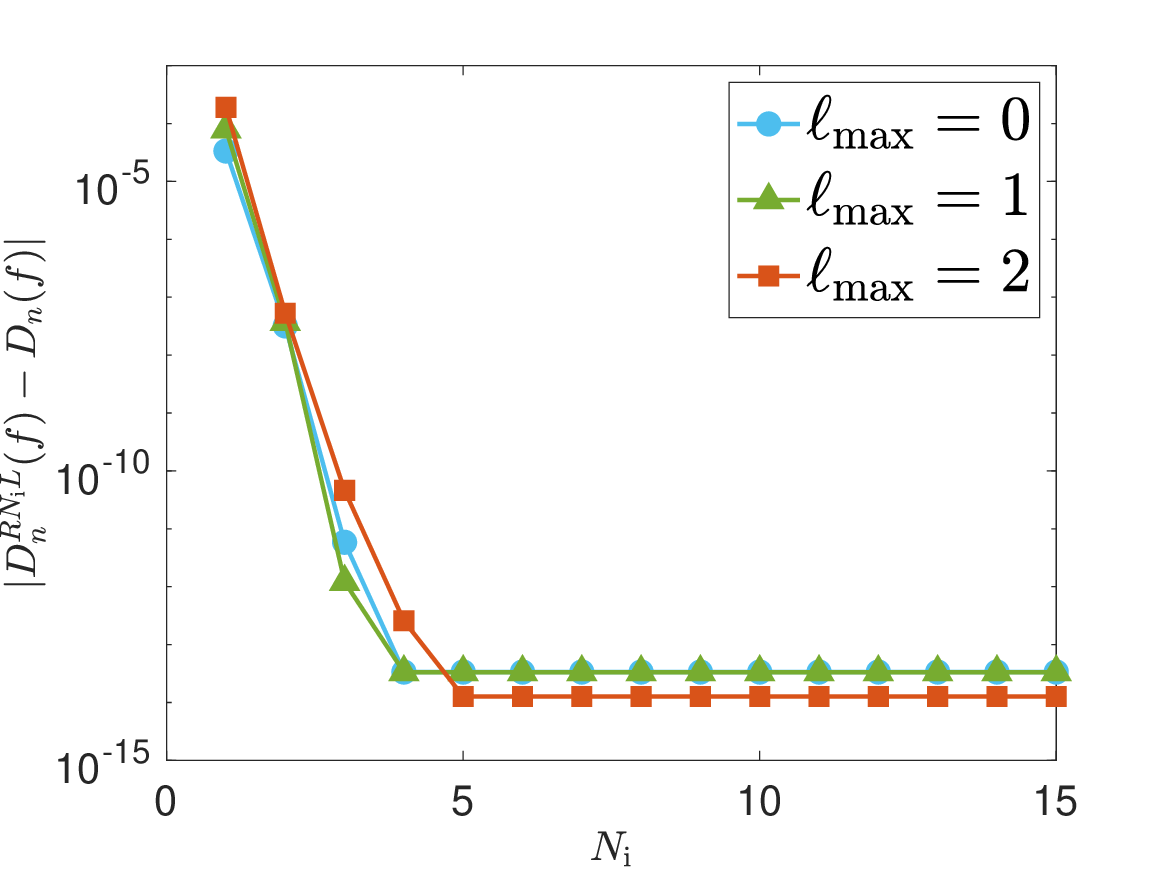}
	\label{fig-Nit_3d}
}
\subfloat[]{
	\includegraphics[width=0.33\textwidth]{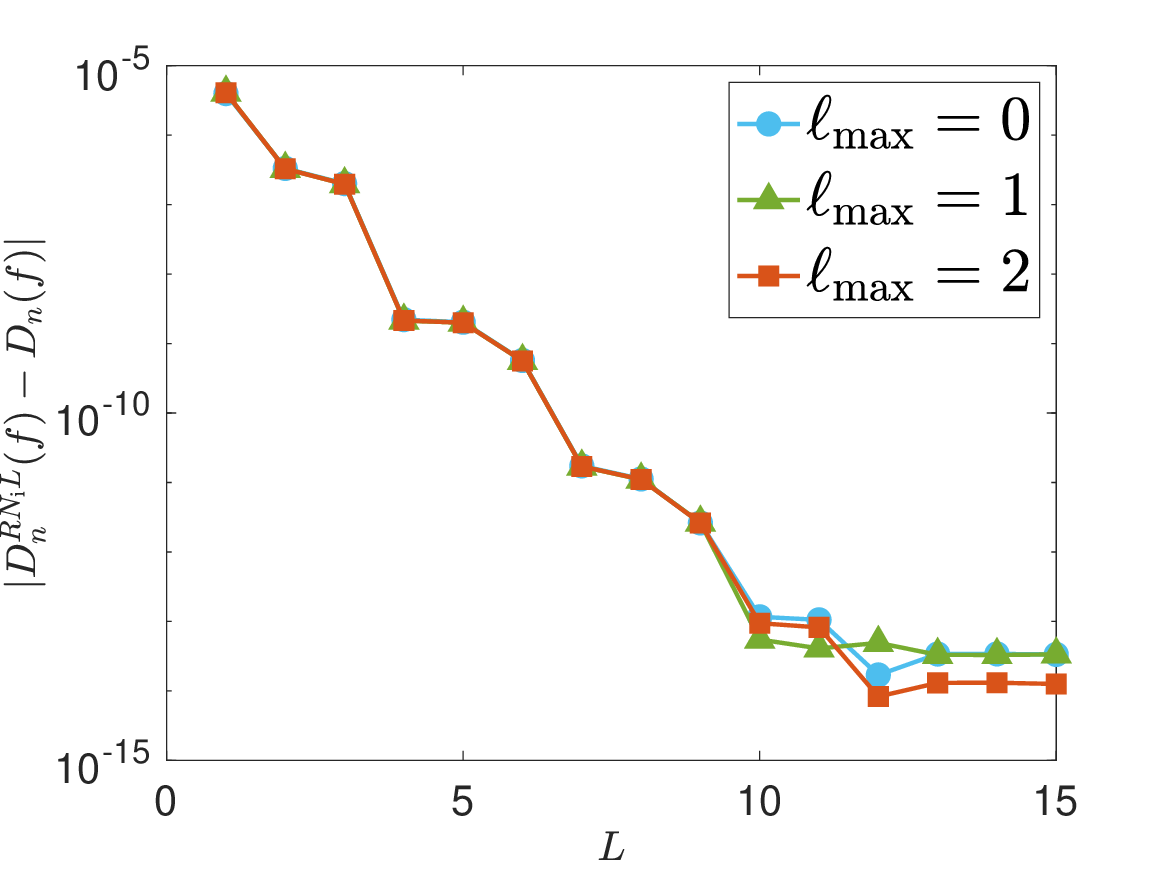}
	\label{fig-Ntr_3d}
}		
\caption{Example 3: Convergence of local DoS. (a) Error w.r.t. size of scattering region $R$; (b) Error w.r.t. number of iterations $\Nit$; (c) Error w.r.t. scattering length of reference $\Ntr$.}	
\end{figure}

\section{Conclusions}
\label{sec:conclusions}

In this paper, we study the numerical approximations of the SPM under within an appropriate screening potential.
We justify the exponential convergence of local DoS with respect to the size of scattering region and the length of scattering path, respectively.
It would be of great interest to study the spectrum distribution of typical disordered systems, such as edge states of surfaces \cite{colbrook2021computing,colbrook2023computing, thicke2021computing}, quasi-periodic systems \cite{massatt2017electronic, wang2022convergence}, and substitutional alloys \cite{ebert11, ruban08} within the MST framework.

\appendix
\renewcommand\thesection{\appendixname~\Alph{section}}

\section{Single-site scattering}
\label{append:single-site}
\renewcommand{\theequation}{A.\arabic{equation}}
\setcounter{equation}{0}

In this appendix, we elaborate on the single-site scattering problem, which can provide the local quantities, i.e., local solutions and $t$-matrix for a given local potential.
Let $n\in\Lambda$ and the local potential $v_n:=V|_{\Omega_n}$ vanishes outside $\Omega_n$, i.e. $v_n(\vr)\equiv0$ for $\vr\notin\Omega_n$.
Consider the following single-site scattering problem with a fixed energy parameter $z\in\C$
\begin{eqnarray}
\label{single-schrodinger}
\Big(-\Delta + v_n(\vr)-z\Big) \zeta_n(\vr;z)= 0 \qquad \bm{r}\in\mathbb{R}^3.
\end{eqnarray}
Note that no (far-field) boundary condition is imposed on this equation, since the solution $\zeta_n$ is severed as a basis within the MST formalism, and will be normalized in some sense below.
For the simplicity of notations, we fix the site index $n\in\Lambda$ and suppress the dependency in this appendix. 
Specifically, $\Omega_n$, $v_n$ and $\zeta_n$ will be abbreviated by $\Omega$, $v$ and $\zeta$, respectively. 
Note that the local potential in an arbitrary-shaped cell $\Omega$ can be considered in the circumscribed ball (still denoted by $\Omega$) with zero-paddings.

For an illustrative purpose, we assume $\Omega$ is an atomic ball centered at the origin with the radius $\mathsf{r}$ and $v$ is spherically symmetric, and  
refer to our previous work \cite[Appendix C]{li2023numerical} for a generalization of cases involving non-spherically symmetric potentials.
For a spherically symmetric potential $v$, the solution $\zeta$ can be written in the angular-momentum representation
\begin{eqnarray}
\label{radial-representation}
\zeta(\vr) = \sum_{\ell m} \zeta_{\ell m}(\vr) = \sum_{\ell m}  \chi_{\ell}(r;z) Y_{\ell m}(\hat{\vr}),
\end{eqnarray}
where $Y_{\ell m}$ are the spherical harmonic functions and the radial part $\chi_{\ell}$ satisfy the radial Schr\"{o}dinger equation
\begin{subequations}
	\label{radial-schrodinger}
	\begin{empheq}[left=\empheqlbrace]{align}
	\label{radial-schrodinger-a}
	\displaystyle\left(\frac{\partial^2}{\partial r^2} + \frac{2}{r} \frac{\partial }{\partial r} -  \frac{\ell(\ell+1)}{r^2} - v(r) +z \right) \chi_{\ell}(r;z)=0 
	& \qquad  0<r<\mathsf{r},
	\\[1ex]
	\label{radial-schrodinger-b}
	\displaystyle\left(\frac{\partial^2}{\partial r^2} + \frac{2}{r} \frac{\partial }{\partial r} -  \frac{\ell(\ell+1)}{r^2} +z \right) \chi_{\ell}(r;z)=0
	& \qquad r>\mathsf{r}.
	\end{empheq}
\end{subequations}
We first study \eqref{radial-schrodinger-b} outside the atomic ball with vanishing potential, which is the well-known spherical Bessel differential equation (see, e.g., \cite{gil07}). 
For given $\ell\in\mathbb{N}$, the solutions could be the spherical Bessel functions $j_{\ell}(\sqrt{z}r)$, the spherical Neumann functions $n_{\ell}(\sqrt{z}r)$, or the spherical Hankel functions $h_{\ell}(\sqrt{z}r)=j_{\ell}(\sqrt{z}r)+i n_{\ell}(\sqrt{z}r)$.
As any two of the above functions can form a complete basis set for the solutions of \eqref{radial-schrodinger-b}, the radial solutions outside the atomic ball can be represented by a linear combination of the spherical Bessel and Hankel functions
\begin{eqnarray}\nonumber
\label{radial-general-eq}
\label{chi_l}
\chi_{\ell}(r;z) = j_{\ell}(\sqrt{z}r) + t_{\ell}(z) h_{\ell}(\sqrt{z}r) \qquad r>\mathsf{r},
\end{eqnarray}
where the normalization is included.
From a physical point of view, $\chi_{\ell}$ can be viewed as a superposition of an incoming wave and a scattering wave,  corresponding to the Bessel and Hankel functions, respectively.
In particular, the $t$-matrix $t_{\ell}$ determines the relationship between these two waves.

From the numerical intuition, the $t$-matrix can be determined by matching the values and derivatives of the radial solution at $r=\mathsf{r}$.
Since multiplying $\chi_{\ell}$ by an arbitrary constant does not change the solution of \eqref{radial-schrodinger}, it is only required to match the ``logarithm derivatives" $\big(\log(\chi_{\ell})\big)'$ at $r=\mathsf{r}$: 
\begin{eqnarray}
\label{log-match}
\frac{\chi_{\ell}(r;z)}{\chi'_{\ell}(r;z)} \bigg|_{r=\mathsf{r}^-} = \frac{j_{\ell}(\sqrt{z}r) + t_{\ell}(z) h_{\ell}(\sqrt{z}r)}{\sqrt{z} j'_{\ell}(\sqrt{z}r) + \sqrt{z} t_{\ell}(z) h'_{\ell}(\sqrt{z}r)}\bigg|_{r=\mathsf{r}^+}.
\end{eqnarray}
Then a direct calculation leads to
\begin{eqnarray}
\label{log-deriv}
t_{\ell}(z)=\frac{\sqrt{z} \chi_{\ell}(\mathsf{r};z) j'_{\ell}(\sqrt{z} \mathsf{r}) -  \chi'_{\ell}(\mathsf{r};z) j_{\ell}(\sqrt{z}\mathsf{r})}
{\chi'_{\ell}(\mathsf{r};z) h_{\ell}(\sqrt{z}\mathsf{r}) - \sqrt{z} \chi_{\ell}(\mathsf{r};z) h'_{\ell}(\sqrt{z} \mathsf{r})} ,
\end{eqnarray}
where we denoted by $\chi'_{\ell}(\mathsf{r};z):=\frac{\partial}{\partial r}\chi_{\ell}(r;z)\big|_{r=\mathsf{r}^-}$ the left-hand derivative of $\chi_{\ell}$ at $r=\mathsf{r}$.
Note that the so-called ``asymptotic problem" \label{asymptotic_problem} (see, e.g., \cite{chen15a} and \cite{singh06}) with $\chi_{\ell}(\mathsf{r};z)=0$ must be avoided such that the conditions \eqref{log-match} and \eqref{log-deriv} can make sense.
We remark that $t_{\ell}$ is well defined as any multiplicative constant involved in the regular solution (finite at the origin) $\chi_{\ell}$ can be canceled in \eqref{log-deriv}.
In addition, the local solutions $\zeta_{\ell m}(\vr;z) =\chi_{\ell}(r;z)Y_{\ell m}(\hat{\vr})$ can be obtained by solving the equation \eqref{radial-schrodinger-a} with the boundary conditions
\begin{eqnarray}\nonumber
\left\{
\begin{array}{l}
\chi_{\ell}(\mathsf{r}; z) =  j_{\ell}(\sqrt{z} \mathsf{r}) + t_{\ell}(z) h_{\ell}(\sqrt{z} \mathsf{r}), \\[1ex]
\displaystyle 
\frac{\partial \chi_{\ell}(r; z)}{\partial r} \Bigg|_{r=\mathsf{r}} = \sqrt{z} j'_{\ell}(\sqrt{z} \mathsf{r}) + \sqrt{z} t_{\ell}(z) h'_{\ell}(\sqrt{z} \mathsf{r}).
\end{array}\right.
\end{eqnarray}

On the other hand, we can define another  solutions $\xi_{\ell m}(\vr;z) =\Xi_{\ell}(r;z)Y_{\ell m}(\hat{\vr})$, whose radial parts also satisfy the equation  \eqref{radial-schrodinger-a}, but with different boundary conditions
\begin{eqnarray}\nonumber
\left\{
\begin{array}{l}
	\Xi_{\ell}(\mathsf{r}; z) =  h_{\ell}(\sqrt{z} \mathsf{r}), \\[1ex]
	\displaystyle 
	\frac{\partial 	\Xi_{\ell}(r; z)}{\partial r} \Bigg|_{r=\mathsf{r}} = \sqrt{z} h'_{\ell}(\sqrt{z} \mathsf{r}).
\end{array}\right.
\end{eqnarray}
The boundary conditions have been shown to result in solutions that usually diverge near the origin, which is referred to {\it irregular} solutions \cite{mavropoulos06, rusanu2011green}.

It is worth noting that all the local quantities can be obtained by solving the single-site scattering problem with a given local potential.
Regardless of the shape and symmetry of the local potentials, the computational complexity of the single scattering problem is relatively small in the MST framework. 
Therefore, we typically assume that all local quantities are computed with sufficient accuracy.

\section{Alternative Cauchy contour}
\label{append:contour}
\renewcommand{\theequation}{B.\arabic{equation}}
\renewcommand{\thefigure}{B.\arabic{figure}}
\setcounter{figure}{0}
\setcounter{equation}{0}

There is an alternative approach to write the contour integral representation of local DoS.
Let $\widetilde{\mathscr{C}}$ be a contour that has a more ``regular" shape but allowed to contain several poles of $f$, which is different from the contour defined in Section \ref{sec:disorder} (see Figure \ref{fig-oval-contour}, where the poles inside the contour are denoted by $z_j,~j=\pm 1, \cdots \pm J$).

\begin{figure}[htbp]
	\centering
	\includegraphics[width=10.cm]{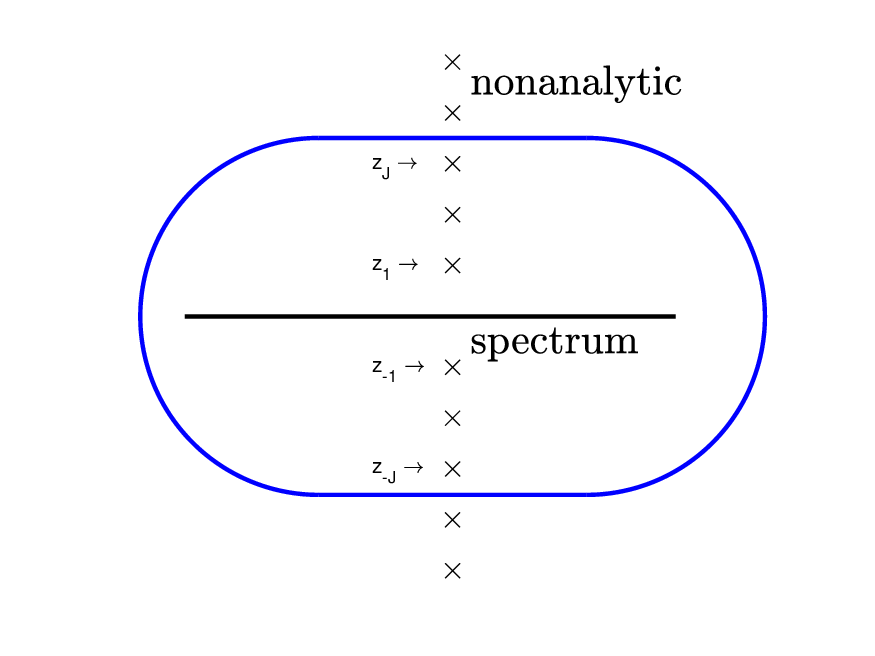}
	\vskip -0.5cm
	\caption{A schematic illustration of quasi-oval contour.}
	\label{fig-oval-contour}
\end{figure}

By using the residue theorem, we have the following representation of local DoS that is equivalent to \eqref{ldos_n} 
\begin{eqnarray}
\label{ldos_2}
\mathscr{D}_n(f) 
= \frac{1}{2\pi i} \int_{\Omega_n} \Bigg( \oint_{\widetilde{\mathscr{C}}} f(z)  G(\vr,\vr;z) \dd z
+ 2\pi i \sum_{j=\pm 1,\cdots,\pm J}  G(\vr,\vr,z_j) \cdot \underset{z=z_j}{{\rm res}}  f(z) \Bigg) \dd\vr,
\end{eqnarray}
where $\underset{z=z_j}{{\rm res}}  f(z)$ stands for the residue of complex function $f$ at isolated singularity $z_j$.
This formula provides a more efficient approach in practical calculations. 
The quasi-oval contour is more ``regular" than the dumbbell-shaped one used in \eqref{ldos_n}, which is therefore much easier to be parameterized and then discretized to perform the numerical quadrature along the contour.
In addition, for the Fermi-Dirac function $f_{\rm FD}$ as the observable function, the integration in \eqref{ldos_2} can be simplified by using $\displaystyle \underset{z=z_j}{{\rm res}}  f_{\rm FD}(z) = -k_{\rm B}T$.

\small
\bibliographystyle{plain}
\bibliography{bib}

\end{document}